\documentclass[]{IEEEtran}

\usepackage[T1]{fontenc}
\usepackage{amsfonts}
\usepackage{amsmath}
\usepackage{amssymb} 
\usepackage{amsthm} 
\usepackage{array}
\usepackage{bm} 
\usepackage[english]{babel}
\usepackage{cases}
\usepackage{cite}
\usepackage{graphicx}
\usepackage[utf8]{inputenc}
\usepackage{mathtools} 

\DeclareSymbolFont{bbold}{U}{bbold}{m}{n}
\DeclareSymbolFontAlphabet{\mathbbold}{bbold}

\def\zetab{\bm{\zeta}}
\def\etab{\bm{\eta}}
\def\thetab{\bm{\theta}}

\def\lambdab{\bm{\lambda}}
\def\mub{\bm{\mu}}

\def\phib{\bm{\phi}}
\def\varphib{\bm{\varphi}}

\def\varPib{\bm{\varPi}}

\def\ab{\bm{a}}

\def\cb{\bm{c}}

\def\fb{\bm{f}}
\def\gb{\bm{g}}
\def\hb{\bm{h}}

\def\tb{\bm{t}}

\def\xb{\bm{x}}

\def\zb{\bm{z}}
\def\Ab{\bm{A}}

\def\Cb{\bm{C}}

\def\Hb{\bm{H}}
\def\Ib{\bm{I}}

\def\Qb{\bm{Q}}
\def\Rb{\bm{R}}

\def\Rbb{\mathbb{R}}
\def\Cbb{\mathbb{C}}

\def\ddiff{\mathrm{d}}
\newcommand*{\diff}{\mathop{}\mathopen{}\ddiff} 
\newcommand*{\eqdef}{\triangleq}
\newcommand*{\semcol}{\mathinner{;}}

\newcommand*{\trsps}{^{\mathrm{T}}}

\newcommand{\given}[1][]{\nonscript\mspace{2mu}#1\vert
    \allowbreak
    \nonscript\mspace{2mu}
    \mathopen{}}
\newcommand{\ttp}[1][]{\mathrel{\overset{#1 P}{\longrightarrow}}}
\newcommand{\ttms}[1][]{\mathrel{\overset{#1 L^{2}}{\longrightarrow}}}
\newcommand{\equivp}[1][]{\overset{#1 P}{\sim}}
\newcommand{\equivms}[1][]{\overset{#1 L^{\mathrlap{2}}}{\sim}}

\newcommand{\attp}{\mathchoice{\underset{N\rightarrow\infty}{\ttp}}{\ttp[\scriptscriptstyle]_{N\rightarrow\infty}}{\ttp_{N\rightarrow\infty}}{\ttp_{N\rightarrow\infty}}}

\newcommand{\aequivp}{\mathchoice{\underset{N\rightarrow\infty}{\equivp}}{\equivp[\scriptscriptstyle]_{N\rightarrow\infty}}{\equivp_{N\rightarrow\infty}}{\equivp_{N\rightarrow\infty}}}

\def\ExpctE{\mathrm{E}}
\makeatletter
\newcommand{\Expct}{\@ifstar\Expct@star\Expct@nostar}
\newcommand{\Expct@star}[2]{%
    \mathop{\ExpctE_{#1}}\mathopen{}\left[#2\right]}
\newcommand{\Expct@nostar}[3][]{%
    \mathop{\ExpctE_{#2}}\mathopen{#1[}#3\mathclose{#1]}}
\makeatother
\def\VarV{\mathrm{V}}
\makeatletter
\newcommand{\Var}{\@ifstar\Var@star\Var@nostar}
\newcommand{\Var@star}[2]{%
    \mathop{\VarV_{#1}}\mathopen{}\left[#2\right]}
\newcommand{\Var@nostar}[3][]{%
    \mathop{\VarV_{#2}}\mathopen{#1[}#3\mathclose{#1]}}
\makeatother
\makeatletter
\newcommand{\indicf}{\@ifstar\indicf@star\indicf@nostar}
\newcommand{\indicf@star}[2]{%
    \mathbbold{1}_{#1}\mathopen{}\left(#2\right)}
\newcommand{\indicf@nostar}[3][]{%
    \mathbbold{1}_{#2}\mathopen{#1(}#3\mathclose{#1)}}
\makeatother
\newcommand{\interval}[5][]{\mathopen{#1#2}#3
    \mathclose{}\mathpunct{},#4
    \mathclose{#1#5}}
\newcommand{\intervalcc}[3][]{\interval[#1]{[}{#2}{#3}{]}}

\newcommand{\intervaloo}[3][]{\interval[#1]{]}{#2}{#3}{[}}
\newcommand{\intervalint}[3][]{\mathopen{#1\{}#2
    ,\ldots,#3
    \mathclose{#1\}}}
\newlength{\eqboxstorage}

\DeclarePairedDelimiter{\abs}{\lvert}{\rvert}
\DeclarePairedDelimiter{\norm}{\lVert}{\rVert}
\DeclarePairedDelimiterX{\scalprod}[2]{\langle}{\rangle}{%
    #1,#2}
\DeclarePairedDelimiterXPP{\Prob}[1]{\Pr}{(}{)}{}{%
    \renewcommand{\given}{\nonscript\,\delimsize\vert
        \allowbreak
        \nonscript\,
        \mathopen{}
    }
    #1}
\DeclarePairedDelimiterX{\enstq}[1]\{\}{%
    \renewcommand{\given}{\nonscript\:\delimsize\vert%
        \allowbreak%
        \nonscript\:%
        \mathopen{}%
    }
    \renewcommand{\semcol}{\nonscript\:%
        ;%
        \allowbreak%
        \nonscript\:%
        \mathopen{}%
    }
#1}
\DeclarePairedDelimiterXPP{\Realpart}[1]{\mathrm{Re}}{[}{]}{}{#1}
\DeclarePairedDelimiterXPP{\Trace}[1]{\mathrm{tr}}{[}{]}{}{#1}

\DeclareMathOperator*{\argmax}{arg\,max}

\theoremstyle{plain}
\newtheorem{theorem}{Theorem}
\theoremstyle{definition}
\newtheorem{definition}{Definition}
\theoremstyle{plain}
\newtheorem{proposition}{Proposition}

\begin{document}
    
\title{Some Results on Tighter Bayesian Lower Bounds on the Mean-Square Error}
\author{
        Lucien~Bacharach\thanks{%
                L. Bacharach and É. Chaumette are with ISAE-Supaéro/DEOS, 10,
avenue Édouard Belin, 31055 Toulouse Cedex 4 Toulouse, France.
            },
        Carsten~Fritsche\thanks{%
                C.~Fritsche is with Linköping University, Dept. of Electrical
Engineering, Linköping, Sweden.
            },
        Umut~Orguner\thanks{%
                U.~Orguner is with Middle East Technical University, Dept. of
Electrical Engineering, Ankara, Turkey.
            }, 
        and Éric~Chaumette\thanks{%
                This work was supported in part by the DGA/AID (2018.60.0072.00.470.75.01) and by the Excellence Center at Link\"{o}ping and Lund in Information Technology (ELLIIT).
            }
        }
\date{July 22, 2019}

\maketitle
       
\begin{abstract}
In random parameter estimation, Bayesian lower bounds (BLBs) for the mean-square error have been noticed to not be tight in a number of cases, even when the sample size, or the signal-to-noise ratio, grow to infinity. In this paper, we study alternative forms of BLBs obtained from a covariance inequality, where the inner product is based on the \textit{a posteriori} instead of the joint probability density function. We hence obtain a family of BLBs, which is shown to form a counterpart at least as tight as the well-known Weiss-Weinstein family of BLBs, and we extend it to the general case of vector parameter estimation. Conditions for equality between these two families are provided. Focusing on the Bayesian Cramér-Rao bound (BCRB), a definition of efficiency is proposed relatively to its tighter form, and efficient estimators are described for various types of common estimation problems, e.g., scalar, exponential family model parameter estimation. Finally, an example is provided, for which the classical BCRB is known to not be tight, while we show its tighter form is, based on formal proofs of asymptotic efficiency of Bayesian estimators. This analysis is finally corroborated by numerical results.
\end{abstract}


\section{Introduction} \label{sec:intro}

\subsection{Background on Bayesian lower bounds on the
MSE}\label{ssec:background_bounds}

Parameter estimation from noisy observations is a fundamental problem arising in many fields such as signal processing, system identification, control theory, communications or economics. As a consequence, many parameter estimation methods and algorithms have been proposed in the literature (see, e.g., \cite{Kay93,VT02}). In order to choose the most suitable method, it is crucial to assess their performance and determine the best achievable accuracy. In the Bayesian framework, the unknown parameter is assumed to be random, with a known \textit{a priori} distribution. A commonly adopted risk function (or estimation performance criterion) is the mean-square error (MSE), for which the Bayes (i.e., optimal) estimator is the posterior mean, i.e., the mean of the \textit{a posteriori} p.d.f. \cite{LC03,Rob07}. This estimator achieves the best accuracy in terms of MSE, and, as a consequence, is commonly referred to as the minimum mean-square error (MMSE) estimator. However, it requires the knowledge of the \textit{a posteriori} p.d.f., which is, except for a few special cases (e.g., conjugate priors),  often difficult to obtain as it involves the evaluation of high-dimensional integrals. They can be computed using Monte-Carlo methods, for instance, and then lead to a high computational cost. Even though one is able to compute them, it is still necessary to assess the estimator's MSE, which increases the computational cost even more. Therefore, it is often necessary to resort to suboptimal approaches, whose loss in accuracy has to be assessed.

Usually, the aforementioned difficulties are overcome by resorting to lower bounds on the MSE that derive from mathematical inequalities. Ideally, these lower bounds are sought to be both computationally tractable and tight, i.e., they should provide insight about the MMSE, as accurately as possible. These lower bounds then make up references with which one can compare the performance of any estimator. For the estimation of random parameters, several Bayesian lower bounds (BLBs) have already been proposed and derived (see, e.g., \cite{VTB07} for an overview). The perhaps most widely used BLB is the Bayesian Cramér-Rao bound (BCRB), as it was the first to be derived \cite{Sch57,Gar59,VT68}, and is also very simple to calculate. Nonetheless, the BCRB turns out to be somewhat optimistic, especially for nonlinear estimation problems, where one often notices performance breakdowns in terms of MSE below a specific signal-to-noise ratio (SNR) or sample size: this phenomenon is referred to as the ``threshold effect'', and is not rendered by the BCRB \cite{RB74,VTB07,RFLRN08}. Yet, determining the appearance conditions of the threshold effect is essential so as to specify the estimators' nominal operating conditions. This has led to considerable work on BLBs, giving rise to two main families: i) the Ziv-Zakaï family, relating to the error probability in binary hypothesis tests \cite{ZZ69,BT74,BSEV97}, and ii) the Weiss-Weinstein family (WWF), deriving from covariance inequalities \cite{Sch57,Gar59,VT68,BZ76,WW88,TT10b}. Some of the bounds among each family make it possible to predict the threshold effect \cite{VTB07} (and references therein), \cite{VRBM14}. Despite this feature, there are many problems for which BLBs are not tight, even in the asymptotic regime, whether for standard problems \cite[pp.~11, 37, 38]{VTB07}, or for dynamic nonlinear filtering \cite{RN05,BV06,XGMM13}.

Recently, a class of BLBs has been derived and shown to be at least as tight as those of the WWF \cite{CF18}. More precisely, it was shown that any lower bound in the WWF implies an alternative form, which is tighter than the standard one. Similarly as the standard form, the tighter form of BLBs is based on a covariance inequality principle, but where the inner product is defined w.r.t. the posterior p.d.f. (instead of the joint p.d.f. for standard BLBs of the WWF). However, this study was limited to the case of scalar parameter estimation. The precision gain of the tighter BLBs had not been assessed until our recent conference paper \cite{BFOC19}, which focused on the case of the BCRB and showed promising results: the tighter BCRB (TBCRB) was shown to be asymptotically tight by simulation, but this was not formally proved.

The main contributions of the present paper with respect to previous work \cite{CF18,BFOC19} are threefold. First, we introduce a more general proof of the covariance inequality leading to the tighter form of BLBs, that encompass the case of vector parameter estimation. Second, general conditions for efficiency are studied and provided for various cases of practical interest: estimation of scalar parameter and exponential family parameter. Third, for the example studied in Section \ref{sec:example}, several asymptotic results are provided and formally proved, about the behavior of Bayesian estimators (like the maximum \textit{a posteriori} (MAP) and the MMSE), in particular their efficiency.

The sequel of this paper is organized as follows. Basic notations, definitions and assumptions used throughout the paper are presented in Section \ref{ssec:notations}. We recall basic results on classical BLBs of the WWF in Section \ref{sec:classical_BLBs}. Their tighter counterparts are presented in Section \ref{sec:TBLBs}, and shown to be indeed tighter. Then, in Section \ref{sec:TBCRB} we focus on the case of the BCRB and provide conditions for efficiency (i.e., attainment of the TBCRB) for various problems. In Section \ref{sec:example}, we study a specific estimation problem, and provide formal proofs of asymptotic efficiency of Bayesian estimators, before illustrating them with numerical results. Finally, concluding remarks are reported in Section \ref{sec:conclusion}.

\subsection{Summary of basic notations, definitions and assumptions} 
\label{ssec:notations}

Throughout the present paper, scalar quantities are denoted by italic letters
(e.g., $a$, $A$), vectors by bold lowercase letters (e.g., $\ab$), and matrices
by bold uppercase letters (e.g., $\Ab$). For some given vector $\ab$, its $i$-th
element is denoted by $a_{i}$, and for some given matrix $\Ab$, its $i$-th row
and $j$-th column element is denoted by $A_{i,j}$. The transpose of a vector
$\ab$ (resp. a matrix $\Ab$) is denoted by $\ab\trsps$ (resp. $\Ab\trsps$). For
some given Hilbert space $\mathcal{H}$ with an inner product
$\scalprod{\cdot\,}{\cdot}$, the orthogonal complement of a subspace
$\mathcal{S}\subset\mathcal{H}$ is denoted by $\mathcal{S}^{\perp}$. The
indicator function of a set $\mathcal{A}$ is denoted by
$\indicf{\mathcal{A}}{\cdot}$.

More specifically, let $\mathcal{X}\subseteq\Rbb^{N}$ be the observation space,
whose elements are random observation vectors denoted by
$\xb\eqdef(x_{1},\ldots,x_{N})\trsps$, and $\Theta\subseteq\Rbb^{K}$ be the
parameter space whose elements are denoted by
$\thetab\eqdef(\theta_{1},\ldots,\theta_{K})\trsps$. 
Let $p(\xb,\thetab)$ denote the joint p.d.f. of $\xb$ and $\thetab$, and
$\mathcal{S}_{\mathcal{X},\Theta}$ its support, i.e.,
$\mathcal{S}_{\mathcal{X},\Theta}\eqdef\enstq[]{(\xb\trsps,\thetab\trsps)\trsps\in\mathcal{X}\times\Theta\given
p(\xb,\thetab)>0}\subseteq\Rbb^{N}\times\Rbb^{K}$. Likewise, let us denote by
$\mathcal{S}_{\Theta}$ the support of the prior p.d.f. $p(\thetab)$ on the one
hand (i.e., $\mathcal{S}_{\Theta}\eqdef\enstq[]{\thetab\in\Theta\given p(\thetab)>0}$), and by $\mathcal{S}_{\mathcal{X}}$ the support of the marginal p.d.f. $p(\xb)$ (i.e., $\mathcal{S}_{\mathcal{X}}\eqdef\enstq[]{\xb\in\mathcal{X}\given p(\xb)>0}$). In addition, let us define $\mathcal{S}_{\Theta\given\xb}$ and $\mathcal{S}_{\mathcal{X}\given\thetab}$ as the supports of the joint p.d.f. $p(\xb,\thetab)$ w.r.t. $\thetab$ and $\xb$ respectively, i.e.,
\begin{itemize}
    \item[i)] for some given $\xb\in\mathcal{S}_{\mathcal{X}}$,
$\mathcal{S}_{\Theta\given\xb}\eqdef\enstq[]{\thetab\in\Theta\given
p(\xb,\thetab)>0}$;
    
    \item[ii)] for some given $\thetab\in\mathcal{S}_{\Theta}$,
$\mathcal{S}_{\mathcal{X}\given\thetab}\eqdef\enstq[]{\xb\in\mathcal{X}\given
p(\xb,\thetab)>0}$.
\end{itemize}
Consequently, one can write
\begin{equation}
    p(\thetab)=\int_{\mathcal{S}_{\mathcal{X}\given\thetab}}
    p(\xb,\thetab)\diff\xb,
    \quad\text{and}\quad
    p(\xb)=\int_{\mathcal{S}_{\Theta\given\xb}}
    p(\xb,\thetab)\diff\thetab.
\end{equation}
From these definitions, the various expectations of a deterministic and
measurable function $\hb:\mathcal{X}\times\Theta\rightarrow\Rbb^{J}$ can be
written as
\begin{IEEEeqnarray}{rCl}
    \IEEEyesnumber\IEEEyessubnumber*
    \Expct{\xb,\thetab}{\hb(\xb,\thetab)}&=&
    \int_{\mathcal{S}_{\mathcal{X},\Theta}}
    \hb(\xb,\thetab)\,p(\xb,\thetab)\diff\xb\diff\thetab, \\
    \Expct{\xb\given\thetab}{\hb(\xb,\thetab)}&=&
    \int_{\mathcal{S}_{\mathcal{X}\given\thetab}}
    \hb(\xb,\thetab)\,p(\xb\given\thetab)\diff\xb, \\
    \Expct{\thetab\given\xb}{\hb(\xb,\thetab)}&=&
    \int_{\mathcal{S}_{\Theta\given\xb}}
    \hb(\xb,\thetab)\,p(\thetab\given\xb)\diff\thetab, \\
    \Expct{\xb}{\hb(\xb,\thetab)}&=&
    \int_{\mathcal{S}_{\mathcal{X}}}
    \hb(\xb,\thetab)\,p(\xb)\diff\xb, \\
    \Expct{\thetab}{\hb(\xb,\thetab)}&=&
    \int_{\mathcal{S}_{\Theta}}
    \hb(\xb,\thetab)\,p(\thetab)\diff\thetab.
\end{IEEEeqnarray}
Accordingly, the variances of $\hb(\xb,\thetab)$ w.r.t. the different distributions appearing above, are respectively denoted by $\Var{\xb,\thetab}{\hb(\xb,\thetab)}=\Expct[\big]{\xb,\thetab}{\bigl(\hb(\xb,\thetab)-\Expct{\xb,\thetab}{\hb(\xb,\thetab)}\bigr)^{2}}$, $\Var{\xb|\thetab}{\hb(\xb,\thetab)}$, $\Var{\thetab|\xb}{\hb(\xb,\thetab)}$, $\Var{\xb}{\hb(\xb,\thetab)}$ and $\Var{\thetab}{\hb(\xb,\thetab)}$.

In addition, throughout the present paper:

\noindent$\bullet$ $\gb:\Theta\rightarrow\Rbb^{L}$ denotes some deterministic,
known function, of which we seek to estimate
$\gb(\thetab)\eqdef(g_{1}(\thetab),\ldots,g_{L}(\thetab))\trsps$. We assume
that, $\forall \ell\in\intervalint{1}{L}$,
$\forall\xb\in\mathcal{S}_{\mathcal{X}}$, 
$g_{\ell}(\cdot)\in\mathcal{L}_{2}(\mathcal{S}_{\Theta\given\xb})$, where
$\mathcal{L}_{2}(\mathcal{S}_{\Theta\given\xb})$ denotes the space of functions
with finite second moments w.r.t. $p(\thetab\given\xb)$, i.e.,
$\Expct{\thetab\given\xb}{g_{\ell}^{2}(\thetab)}<+\infty$.

\noindent$\bullet$ $\widehat{\gb}:\mathcal{X}\rightarrow\Rbb^{L}$ denotes some
estimator of $\gb(\thetab)$. We assume that, $\forall
\ell\in\intervalint{1}{L}$,
$\widehat{g}_{\ell}(\cdot)\in\mathcal{L}_{2}(\mathcal{S}_{\mathcal{X}})$, where
$\mathcal{L}_{2}(\mathcal{S}_{\mathcal{X}})$ denotes the space of functions with
finite second moments w.r.t. $p(\xb)$, i.e.,
$\Expct{\xb}{\widehat{g}_{\ell}^{2}(\xb)}<+\infty$.

\noindent$\bullet$ $\varphi:\mathcal{X}\times\Theta\rightarrow\Rbb$ denotes some
deterministic, known function. We assume that
$\varphi(\cdot)\in\mathcal{L}_{2}(\mathcal{S}_{\mathcal{X},\Theta})$, where
$\mathcal{L}_{2}(\mathcal{S}_{\mathcal{X},\Theta})$ denotes the space of
functions with finite second moments w.r.t. $p(\xb,\thetab)$, i.e.,
$\Expct{\xb,\thetab}{\varphi^{2}(\xb,\thetab)}<+\infty$.

\noindent$\bullet$ The inner product of two functions
$\zeta(\cdot),\xi(\cdot)\in\mathcal{L}_{2}(\mathcal{S}_{\mathcal{X},\Theta})$ is
defined by
\begin{equation}
    \scalprod{\zeta(\xb,\thetab)}{\xi(\xb,\thetab)}\eqdef
    \Expct{\xb,\thetab}{\zeta(\xb,\thetab)\,\xi(\xb,\thetab)},
    \label{eqn:innerprod1}
\end{equation}
and the natural norm based on it is denoted by $\norm{.}$.

\noindent$\bullet$ For some family of functions (or vector function)
$\varphib(\cdot)\eqdef(\varphi_{1}(\cdot),\ldots,\varphi_{M}(\cdot))\trsps\in\mathcal{L}_{2}^{M}(\mathcal{S}_{\mathcal{X},\Theta})$
with finite second moments w.r.t. $p(\xb,\thetab)$, we denote by
$\mathcal{S}_{\varphib}$ its linear span, 
that is 
\begin{equation}
    \mathcal{S}_{\varphib}\eqdef\mathrm{span}(\varphib(\cdot))=\enstq[\big]{
        \lambdab\trsps\varphib(\cdot)\given \lambdab\in\Rbb^{M}
    }.
\end{equation}
The orthogonal complement of $\mathcal{S}_{\varphib}$ in
$\mathcal{L}_{2}(\mathcal{S}_{\mathcal{X},\Theta})$ for the inner product
\eqref{eqn:innerprod1} is denoted by $\mathcal{S}_{\varphib}^{\perp}$, i.e.,
\begin{IEEEeqnarray}{/l}
    \mathcal{S}_{\varphib}^{\perp}\eqdef
    \enstq[\Big]{
        \zeta(\cdot)\in\mathcal{L}_{2}(\mathcal{S}_{\mathcal{X},\Theta}) \given
        \Expct{\xb,\thetab}{\zeta(\xb,\thetab)\,\varphib(\xb,\thetab)}=\bm{0}}.
\end{IEEEeqnarray}

\noindent$\bullet$ For some subspace $\mathcal{A}$ of
$\mathcal{L}_{2}(\mathcal{S}_{\mathcal{X},\Theta})$, let $\varPi_{\mathcal{A}}:
\mathcal{L}_{2}(\mathcal{S}_{\mathcal{X},\Theta}) \rightarrow
\mathcal{L}_{2}(\mathcal{S}_{\mathcal{X},\Theta})$ denote the orthogonal
projector onto $\mathcal{A}$ based on the inner product \eqref{eqn:innerprod1}.
Thus, any function
$f(\cdot)\in\mathcal{L}_{2}(\mathcal{S}_{\mathcal{X},\Theta})$ can be decomposed
as
\begin{equation}
    f(\xb,\thetab) = \varPi_{\mathcal{A}}(f)(\xb,\thetab) +
\varPi_{\mathcal{A}^{\perp}}(f)(\xb,\thetab).
    \label{eqn:proj_decomp}
\end{equation}
Accordingly, we denote by $\varPib_{\mathcal{A}}(\fb)$ and
$\varPib_{\mathcal{A}^{\perp}}(\fb)$ the element-wise orthogonal projections
of a vector function
$\fb(\cdot)\eqdef(f_{1}(\cdot),\ldots,f_{J}(\cdot))\trsps\in\mathcal{L}_{2}^{J}(\mathcal{S}_{\mathcal{X},\Theta})$
onto $\mathcal{A}$ and $\mathcal{A}^{\perp}$ respectively, i.e.,
$\varPib_{\mathcal{A}}(\fb)\eqdef
(\varPi_{\mathcal{A}}(f_{1}),\ldots,\varPi_{\mathcal{A}}(f_{J}))\trsps$.

\medskip
Using these notations and definitions, we now derive two classes of Bayesian
information inequalities, the first of which leads to the so-called
Weiss-Weinstein family (WWF) of Bayesian lower bounds (BLBs) 
\cite{WW88,TT10b,CRE17}, while the second, introduced more recently \cite{CF18},
leads to their tighter counterparts (tighter Bayesian lower bounds, TBLBs).

\section{Classical Bayesian lower bounds} \label{sec:classical_BLBs}

In this section, we recall the basic form of the covariance inequality, which
leads to a general class of Bayesian lower bounds on the global mean-square
error of any Bayes estimator. In particular, it includes the fairly well-known
Weiss-Weinstein family of BLBs \cite{WW88,TT10b,CRE17}.

\subsection{Background on covariance inequality}
\label{ssec:background_cov_ineq}


A basic form of the covariance inequality can be stated as follows (a proof is
given to enable an easier comparison between the lower bounds from the present
section and those from Section \ref{sec:TBLBs}).

\begin{theorem}[Covariance inequality] \label{thm:cov_ineq}
    Let
$\gb(\cdot)\eqdef(g_{1}(\cdot),\ldots,g_{L}(\cdot))\trsps\in\mathcal{L}_{2}^{L}(\mathcal{S}_{\mathcal{X},\Theta})$
be some vector function with finite second moments w.r.t. $p(\xb,\thetab)$, and
$\varphib(\cdot)\eqdef (\varphi_{1}(\cdot),\ldots,\varphi_{M}(\cdot))\trsps
\in\mathcal{L}_{2}^{M}(\mathcal{S}_{\mathcal{X},\Theta})$ be some family of
linearly independent functions with finite second moments w.r.t.
$p(\xb,\thetab)$ as well. Then, for any vector function
$\zetab(\cdot)\eqdef(\zeta_{1}(\cdot),\ldots,\zeta_{L}(\cdot))\trsps\in(\mathcal{S}_{\varphib}^{\perp})^{L}$,
    \begin{equation}
        \Qb_{(\gb-\zetab)}\succeq
        \Rb_{\gb\varphib}\:\Qb_{\varphib}^{-1}\:\Rb_{\gb\varphib}\trsps,
        \label{eqn:cov_ineq}
    \end{equation}
    where $\Qb_{(\gb-\zetab)}$ is the $L\times L$ matrix defined by
    \begin{equation}
        \Qb_{(\gb-\zetab)}\eqdef
        \Expct[\big]{\xb,\thetab}{(\gb(\xb,\thetab)-\zetab(\xb,\thetab))\cdot(\gb(\xb,\thetab)-\zetab(\xb,\thetab))\trsps},
    \end{equation}
    $\Rb_{\gb\varphib}$ is the $L\times M$ matrix defined by
    \begin{equation}
        \Rb_{\gb\varphib}\eqdef
        \Expct[\big]{\xb,\thetab}{\gb(\xb,\thetab)\,\varphib\trsps(\xb,\thetab)},
        \label{eqn:def_Rgphi}
    \end{equation}
    $\Qb_{\varphib}$ is the $M\times M$ matrix defined by
    \begin{equation}
        \Qb_{\varphib}\eqdef
        \Expct[\big]{\xb,\thetab}{\varphib(\xb,\thetab)\,\varphib\trsps(\xb,\thetab)},
        \label{eqn:def_Qphi}
    \end{equation}
    and the inequality sign ``$\succeq$'' in \eqref{eqn:cov_ineq} means that the
difference between the left and the right side is a positive semi-definite
matrix.
\end{theorem}

\begin{IEEEproof}
    Let $\ab\in\Rbb^{L}$ be any vector, and
$\epsilon(\xb,\thetab)\eqdef\ab\trsps(\gb(\xb,\thetab)-\zetab(\xb,\thetab))$. From
\eqref{eqn:proj_decomp}, $\epsilon(\xb,\thetab)$ can be written as
    \begin{IEEEeqnarray*}{rCl}
        \epsilon(\xb,\thetab) &=& \varPi_{\mathcal{S}_{\varphib}}(\epsilon)(\xb,\thetab) +
        \varPi_{\mathcal{S}_{\varphib}^{\perp}}(\epsilon)(\xb,\thetab) \\
        &=&
        \ab\trsps\varPib_{\mathcal{S}_{\varphib}}(\gb-\zetab)(\xb,\thetab)+\ab\trsps\varPib_{\mathcal{S}_{\varphib}^{\perp}}(\gb-\zetab)(\xb,\thetab)
        \\
        &=&
        \ab\trsps\varPib_{\mathcal{S}_{\varphib}}(\gb)(\xb,\thetab)+\ab\trsps\bigl[\varPib_{\mathcal{S}_{\varphib}^{\perp}}(\gb)(\xb,\thetab)-\zetab(\xb,\thetab)\bigr]
        \\
        \IEEEyesnumber\label{eqn:e_projections}
    \end{IEEEeqnarray*}
    since $\zetab(\cdot)\in(\mathcal{S}_{\varphib}^{\perp})^{L}$. Consequently,
by the Pythagorean theorem, we have
    \begin{IEEEeqnarray*}{rCl}
        \norm{\epsilon(\xb,\thetab)}^{2} &=&
        \norm{\ab\trsps\varPib_{\mathcal{S}_{\varphib}}(\gb)(\xb,\thetab)}^{2} 
        \\
        &&{+}\>
        \norm[\big]{\ab\trsps\bigl[\varPib_{\mathcal{S}_{\varphib}^{\perp}}(\gb)(\xb,\thetab)-\zetab(\xb,\thetab)\bigr]}^{2}.
        \IEEEyesnumber\label{eqn:pythagoras1}
    \end{IEEEeqnarray*}
    Since
$\varPib_{\mathcal{S}_{\varphib}}(\gb)(\xb,\thetab)\in(\mathcal{S}_{\varphib})^{L}$,
there exists a matrix $\Cb\in\Rbb^{M\times L}$ such that
$\varPib_{\mathcal{S}_{\varphib}}(\gb)(\xb,\thetab)=\Cb\trsps\varphib(\xb,\thetab)$.
In addition,
$\gb(\xb,\thetab)-\varPib_{\mathcal{S}_{\varphib}}(\gb)(\xb,\thetab)=\varPib_{\mathcal{S}_{\varphib}^{\perp}}(\gb)(\xb,\thetab)\in(\mathcal{S}_{\varphib}^{\perp})^{L}$,
so it satisfies the normal equation
    \begin{equation}
        \Expct[\big]{\xb,\thetab}{(\gb(\xb,\thetab)-\Cb\trsps\varphib(\xb,\thetab))\,\varphib\trsps(\xb,\thetab)}=\bm{0},
    \end{equation}
    which leads to $\Cb=\Qb_{\varphib}^{-1}\Rb_{\gb\varphib}\trsps$, and
$\varPib_{\mathcal{S}_{\varphib}}(\gb)(\xb,\thetab)=
\Rb_{\gb\varphib}\Qb_{\varphib}^{-1}\,\varphib(\xb,\thetab)$. Noticing that
$\norm{\epsilon(\xb,\thetab)}^{2}=\ab\trsps\Qb_{(\gb-\zetab)}\ab$ and
$\norm{\ab\trsps\varPib_{\mathcal{S}_{\varphib}}(\gb)(\xb,\thetab)}^{2}=
\ab\trsps\Rb_{\gb\varphib}\Qb_{\varphib}^{-1}\Rb_{\gb\varphib}\trsps\ab$
    %
    %
    leads to the following equivalent form of \eqref{eqn:pythagoras1}:
    \begin{IEEEeqnarray*}{'rCl}
        \ab\trsps\Qb_{(\gb-\zetab)}\ab &=&
        \ab\trsps\Rb_{\gb\varphib}\Qb_{\varphib}^{-1}\Rb_{\gb\varphib}\trsps\ab
		\\
        &&{+}\>
        \norm[\big]{\ab\trsps\bigl[\varPib_{\mathcal{S}_{\varphib}^{\perp}}(\gb)(\xb,\thetab)-\zetab(\xb,\thetab)\bigr]}^{2}.
        \IEEEyesnumber\label{eqn:pythagoras1bis}
    \end{IEEEeqnarray*}
    It implies the inequality
    $\ab\trsps\Qb_{(\gb-\zetab)}\ab\geq\ab\trsps\Rb_{\gb\varphib}\Qb_{\varphib}^{-1}\Rb_{\gb\varphib}\trsps\ab$,
    for any $\ab\in\Rbb^{L}$, hence the covariance inequality \eqref{eqn:cov_ineq}.
\end{IEEEproof}

\subsection{A general class of BLBs} \label{ssec:general_class_BLBs}

If we set $\gb(\xb,\thetab)\eqdef \gb(\thetab)$ the quantity to be estimated,
and $\zetab(\xb,\thetab)\eqdef
\Expct[]{\thetab\given\xb}{\gb(\thetab)}=\widehat{\gb}^{\mathsf{MMSE}}(\xb)$ the
posterior mean of $\gb(\thetab)$ (that is the MMSE estimator), then, provided
that $\Expct{\thetab\given\xb}{\gb(\thetab)} \in
(\mathcal{S}_{\varphib}^{\perp})^{L}$, the inequality \eqref{eqn:cov_ineq} becomes
\begin{equation}
    \mathbf{MSE}(\widehat{\gb}^{\mathsf{MMSE}})\succeq
    \Rb_{\gb\varphib}\:\Qb_{\varphib}^{-1}\:\Rb_{\gb\varphib}
    \label{eqn:LB_MMSE}
\end{equation}
%
where, for any estimator $\widehat{\gb}(\cdot)$,
$\mathbf{MSE}(\widehat{\gb})\eqdef\Qb_{(\gb-\widehat{\gb})}=\Expct[\big]{\xb,\thetab}{(\gb(\thetab)-\widehat{\gb}(\xb))\cdot(\gb(\thetab)-\widehat{\gb}(\xb))\trsps}$
denotes the mean-square error matrix of $\widehat{\gb}(\cdot)$. Moreover, since,
for any estimator $\widehat{\gb}(\cdot)
\in\mathcal{L}_{2}(\mathcal{S}_{\mathcal{X}})$,
$\mathbf{MSE}(\widehat{\gb})\succeq\mathbf{MSE}(\widehat{\gb}^{\mathsf{MMSE}})$,
any lower bound on the MSE of $\widehat{\gb}^{\mathsf{MMSE}}(\cdot)$ is also a
lower bound on the MSE of $\widehat{\gb}(\cdot)$, whether $\widehat{\gb}(\cdot)$
lies in $(\mathcal{S}_{\varphib}^{\perp })^{L}$ or not. It is worth
noticing that a sufficient condition for \eqref{eqn:LB_MMSE} to hold is
$\Expct{\thetab\given\xb}{\gb(\thetab)}\in (\mathcal{S}_{\varphib}^{\perp })^{L}$, or
equivalently
\begin{equation}
\Expct*{\xb}{
    \Expct{\thetab\given\xb}{\gb(\thetab)}
    \Expct{\thetab\given\xb}{\varphib(\xb,\thetab)}
} =\bm{0}.
\label{eqn:LB_existence_condition}
\end{equation}
A well known subset of the Euclidean space $\mathcal{L}_{2}\left(
\mathcal{S}_{\mathcal{X},\Theta }\right)$ satisfying
\eqref{eqn:LB_existence_condition} is $\mathcal{H}_{\phi}$ defined as
\cite[(1)]{WW88}, \cite[(6)]{TT10b}, \cite{WW14}:
\begin{IEEEeqnarray}{'l}
\mathcal{H}_{\phi }=\bigl\{
    \phi(\xb,\thetab)\in\mathcal{L}_{2}(\mathcal{S}_{\mathcal{X},\Theta})
    \given[\big]
    \Expct{\thetab\given\xb}{\phi(\xb,\thetab)} = 0 \IEEEnonumber\\
    \hspace{15em}
    \text{for a.e. }\xb\in \mathcal{X}
\bigr\},
\label{eqn:def_Hphi}    
\end{IEEEeqnarray}
which is the subset of the Euclidean space
$\mathcal{L}_{2}(\mathcal{S}_{\mathcal{X},\Theta})$ orthogonal to any
$\widehat{\gb}(\xb)\in\mathcal{L}_{2}(\mathcal{S}_{\mathcal{X}})$. Therefore,
there may exist vector functions $\varphib (\xb,\thetab)$ that do not lie in
$\mathcal{H}_{\phi }$ but satisfy \eqref{eqn:LB_existence_condition}; such a
function $\varphib(\xb,\thetab)$ could possibly lead to tighter BLBs since it is
required to be orthogonal only to a single function of
$\mathcal{L}_{2}(\mathcal{S}_{\mathcal{X}})$, that is
$\Expct{\thetab\given\xb}{\gb(\thetab)}$.

\subsection{The Weiss-Weinstein class of BLBs} \label{ssec:WWF}

As stated in \cite{TT10b}, the Weiss-Weinstein class of BLBs is obtained via
projection of $\gb(\thetab)$ on $\mathcal{S}_{\phib}$, that is a closed subspace
of $\mathcal{H}_{\phi}$ made up of linear combinations of elements in
$\mathcal{H}_{\phi}$:
\begin{equation}
    \mathcal{S}_{\phib}\eqdef
    \enstq[\big]{
        \lambdab\trsps\phib(\cdot)\given
        \lambdab\in\Rbb^{M}
    },
\end{equation}
where
$\phib(\cdot)\eqdef(\phi_{1}(\cdot),\ldots,\phi_{M}(\cdot))\trsps\in\mathcal{H}_{\phi}^{M}$.
Then, any estimator $\widehat{\gb}(\xb)$ lies in $\mathcal{S}_{\phib}^{\perp}$,
and setting $\gb(\xb,\thetab)=\gb(\thetab)$ and
$\zetab(\xb,\thetab)\eqdef\widehat{\gb}(\xb)$ in \eqref{eqn:pythagoras1} yields,
for any $\ab\in\Rbb^{L}$, 
\begin{equation}
    \norm{\ab\trsps(\gb(\thetab)-\widehat{\gb}(\xb))}^{2}\geq
    \norm{\ab\trsps\varPib_{\mathcal{S}_{\phib}}(\gb)(\xb,\thetab)}^{2},
\end{equation}
which leads to 
\begin{equation}
    \mathbf{MSE}(\widehat{\gb}) \succeq
    \Rb_{\gb\phib}\:\Qb_{\phib}^{-1}\,\Rb_{\gb\phib}\trsps.
    \label{eqn:LB_MSE_WWF}
\end{equation}
This inequality is referred to as the Weiss-Weinstein inequality in the
following. 

\section{A general class of tighter Bayesian lower bounds} \label{sec:TBLBs} 

\subsection{A tighter version of the covariance inequality}
\label{ssec:tighter_bayes_info_ineq}

We first provide a tighter version of the covariance inequality \eqref{eqn:cov_ineq}. Since, for any function $\epsilon(\cdot)\in\mathcal{L}_{2}(\mathcal{S}_{\mathcal{X},\Theta})$,
\begin{equation}
\norm{\epsilon(\xb,\thetab)}^{2}=
\Expct[\big]{\xb}{\Expct[\big]{\thetab|\xb}{\epsilon^{2}(\xb,\thetab)}},
\label{eqn:split_MSE}
\end{equation}
it is possible to rewrite all the results from Section \ref{ssec:background_cov_ineq} for some given $\xb\in\mathcal{S}_{\mathcal{X}}$, with respect to the posterior distribution $p(\thetab\given\xb)$. In other words, \eqref{eqn:split_MSE} suggests to derive a lower bound on the posterior MSE $\Expct[\big]{\thetab|\xb}{\epsilon^{2}(\xb,\thetab)}$, and then average it w.r.t. $\xb$, which leads to another bound on the global MSE $\Expct[\big]{\xb,\thetab}{\epsilon^{2}(\xb,\thetab)}$. In the next theorem (Theorem \ref{thm:tighter_cov_ineq}), we show that the bound so obtained is at least as tight as the standard one (from the previous section). It is done by noticing that, geometrically, the proposed bound corresponds to the hypotenuse of a right-angled triangle, while the standard bound corresponds to one of the two other sides of this triangle (see Figure \ref{fig:geometric_interpretation}).

In order to state the theorem, let us introduce the subspace $\mathcal{W}_{2}(\mathcal{S}_{\mathcal{X},\Theta})$ of $\mathcal{L}_{2}(\mathcal{S}_{\mathcal{X},\Theta})$, defined as
\begin{IEEEeqnarray*}{'ll}
\mathcal{W}_{2}(\mathcal{S}_{\mathcal{X},\Theta})\eqdef
\Bigl\{
    &\zeta(\cdot)\in\mathcal{L}_{2}(\mathcal{S}_{\mathcal{X},\Theta})\given[\Big]
    \\
    &\zeta(\xb,\cdot)\in\mathcal{L}_{2}(\mathcal{S}_{\Theta|\xb})
        \text{ for a.e. }\xb\in\mathcal{S}_{\mathcal{X}}
        \Bigr\},
\IEEEyesnumber\label{eqn:def_W2}
\end{IEEEeqnarray*}
and the inner product of two functions $\xi(\cdot),\zeta(\cdot)\in\mathcal{W}_{2}(\mathcal{S}_{\mathcal{X},\Theta})$, given some $\xb\in\mathcal{S}_{\mathcal{X}}$, defined by
\begin{equation}
\scalprod{\xi(\xb,\thetab)}{\zeta(\xb,\thetab)}_{|\xb}\eqdef
\Expct[\big]{\thetab|\xb}{\xi(\xb,\thetab)\,\zeta(\xb,\thetab)}.
\label{eqn:innerprod2}
\end{equation}
Accordingly, the related norm is denoted by $\norm{\cdot}_{\given\xb}$. As a consequence, for some given $\xb\in\mathcal{X}$, let
$\mathcal{S}_{\varphib\given\xb}$ denote the linear span of a family of functions $\varphi_{1}(\xb,\cdot),\ldots,\varphi_{M}(\xb,\cdot)$, where $\varphi_{1}(\cdot),\ldots,\varphi_{M}(\cdot)\in\mathcal{W}_{2}
(\mathcal{S}_{\mathcal{X},\Theta})$, and
let $\mathcal{S}_{\varphib\given\xb}^{\perp}$ denote its orthogonal complement for
the inner product given $\xb$ \eqref{eqn:innerprod2}, i.e.,
\begin{equation}
    \mathcal{S}_{\varphib\given\xb}^{\perp}\eqdef\enstq[\Big]{
        \zeta(\cdot)\in\mathcal{W}_{2}(\mathcal{S}_{\mathcal{X},\Theta})
        \given
        \Expct[\big]{\thetab|\xb}{\zeta(\xb,\thetab)\,\varphib(\xb,\thetab)}
        =\bm{0}
    }.
\end{equation}
Finally, let us denote by $\widetilde{\mathcal{S}}_{\varphib}^{\perp}$ the subspace of
$\mathcal{S}_{\varphib}^{\perp}$ defined by
\begin{equation}
    \widetilde{\mathcal{S}}_{\varphib}^{\perp}\eqdef\enstq[\Big]{
        \zeta(\cdot)\in\mathcal{W}_{2}(\mathcal{S}_{\mathcal{X},\Theta})
        \given
        \zeta(\xb,\cdot)\in\mathcal{S}_{\varphib|\xb}^{\perp}
        \text{ for a.e. } \xb\in\mathcal{X}
    }.
\end{equation}
Using these notations, we can now state the main theorem of this section.

\begin{figure}[t!]
    \centering
    \includegraphics[width=\linewidth]{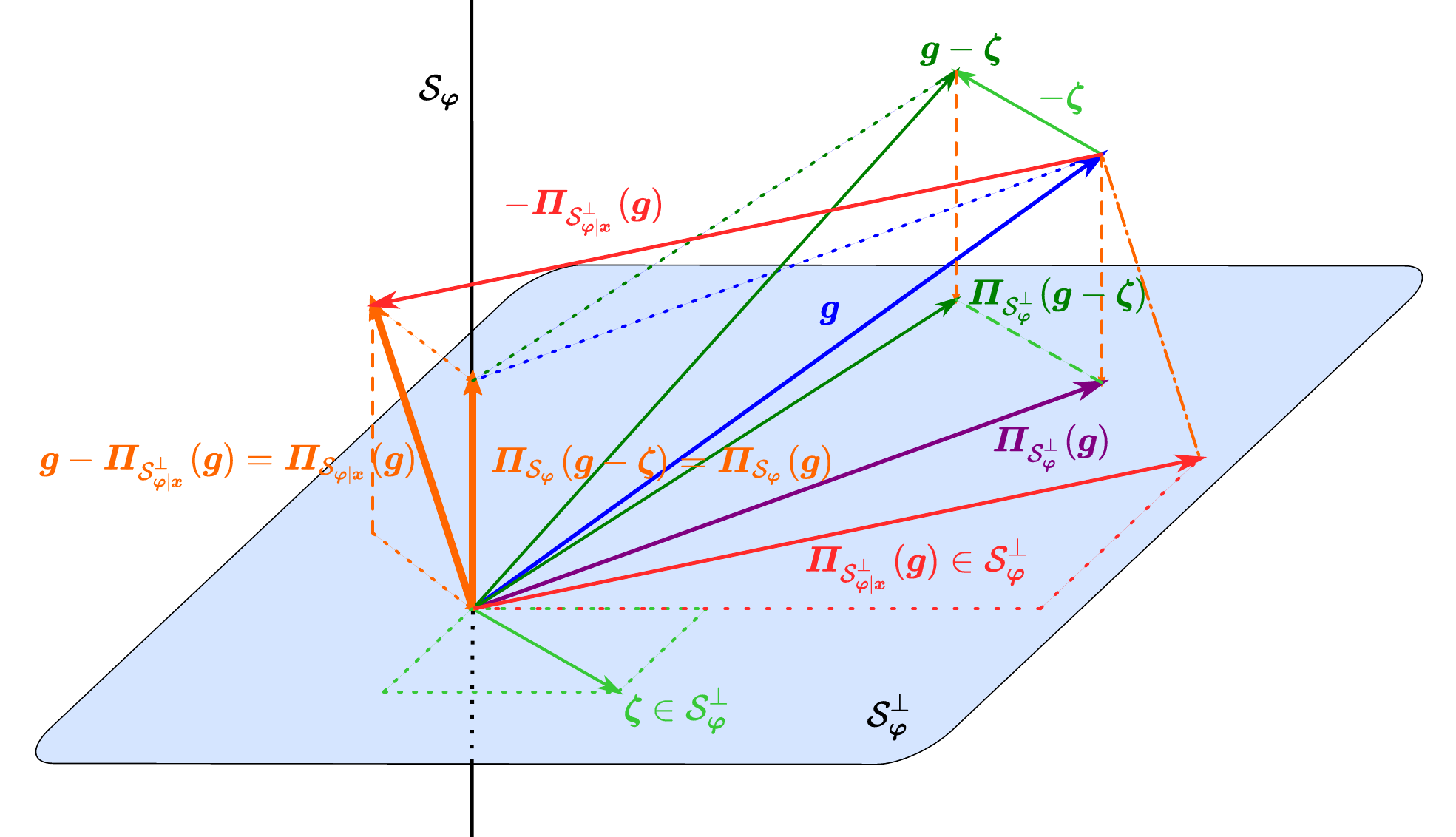}
    \caption{
        Geometrical interpretation of the presented lower bounds: the standard one corresponds to the (squared norm of the) vertical thick orange arrow $\varPib_{\mathcal{S}_{\varphib}}(\gb)$ (or equivalently the vertical orange dashed (`--~--') segments), while the tighter bound corresponds to the (squared norm of the) slanted thick orange arrow $\varPib_{\mathcal{S}_{\varphib\given\xb}}(\gb)$ (or equivalently the slanted orange ``dash-dot'' (`--~$\cdot$') segment).
    }
    \label{fig:geometric_interpretation}
\end{figure}

\begin{theorem}[Tighter covariance inequality] \label{thm:tighter_cov_ineq}
    Let
    $\gb(\cdot)\eqdef(g_{1}(\cdot),\ldots,g_{L}(\cdot))\trsps\in\bigl(\mathcal{W}_{2}(\mathcal{S}_{\mathcal{X},\Theta})\bigr)^{L}$
    be some family of $L$ functions, and $\varphib(\cdot)\eqdef (\varphi_{1}(\cdot),\ldots,\varphi_{M}(\cdot))\trsps \in\bigl(\mathcal{W}_{2}(\mathcal{S}_{\mathcal{X},\Theta})\bigr)^{M}$ be some family of $M$ linearly independent functions. Then, for any family of $L$ functions $\zetab(\cdot)\eqdef(\zeta_{1}(\cdot),\ldots,\zeta_{L}(\cdot))\trsps
    \in(\widetilde{\mathcal{S}}_{\varphib}^{\perp})^{L}$,
    \begin{equation}
        \Qb_{(\gb-\zetab)}\succeq
        \Expct[\big]{\xb}{\Rb_{\gb\varphib\given\xb}\:
        \Qb_{\varphib\given\xb}^{-1}\,
        \Rb_{\gb\varphib\given\xb}\trsps}\succeq
        \Rb_{\gb\varphib}\:\Qb_{\varphib}^{-1}\,\Rb_{\gb\varphib}\trsps,
        \label{eqn:tighter_cov_ineq}
    \end{equation}
    where $\Rb_{\gb\varphib\given\xb}$ is the $L\times M$ matrix defined by
    \begin{equation}
        \Rb_{\gb\varphib\given\xb}\eqdef
        \Expct[\big]{\thetab\given\xb}{\gb(\xb,\thetab)\,\varphib\trsps(\xb,\thetab)},
        \label{eqn:def_Rgphi_x}
    \end{equation}
    and $\Qb_{\varphib\given\xb}$ is the $M\times M$ matrix defined by
    \begin{equation}
        \Qb_{\varphib\given\xb}\eqdef
        \Expct[\big]{\thetab\given\xb}{\varphib(\xb,\thetab)\,\varphib\trsps(\xb,\thetab)}.
        \label{eqn:def_Qphi_x}
    \end{equation}
\end{theorem}

\begin{IEEEproof}
    For any arbitrary vector $\ab\in\Rbb^{L}$ and some given $\xb\in\mathcal{S}_{\mathcal{X}}$, let us decompose $\epsilon(\xb,\thetab)\eqdef\ab\trsps(\gb(\xb,\thetab)-\zetab(\xb,\thetab))$ into
    \begin{IEEEeqnarray*}{rCl}
    \epsilon(\xb,\thetab)
    &=& \varPi_{\mathcal{S}_{\varphib|\xb}}(\epsilon)(\xb,\thetab) +
    \varPi_{\mathcal{S}_{\varphib|\xb}^{\perp}}(\epsilon)(\xb,\thetab) \\
    &=& \ab\trsps\varPib_{\mathcal{S}_{\varphib\given\xb}}(\gb)(\xb,\thetab)\\
    & & {+}\>\ab\trsps\bigl[\varPib_{\mathcal{S}_{\varphib\given\xb}^{\perp}}(\gb)
    (\xb,\thetab)-\zetab(\xb,\thetab)\bigr],
    \IEEEyesnumber\label{eqn:e_projections_2}
    \end{IEEEeqnarray*}
    as in \eqref{eqn:e_projections}.
%
%
    By the Pythagorean theorem, we consequently have
    \begin{IEEEeqnarray*}{'rCl}
        \norm{\epsilon(\xb,\thetab)}_{\given\xb}^{2} &=&
        \norm{\ab\trsps\varPib_{\mathcal{S}_{\varphib\given\xb}}(\gb)(\xb,\thetab)}_{\given\xb}^{2}
         \\
        &&+\>
        \norm[\big]{\ab\trsps\bigl[\varPib_{\mathcal{S}_{\varphib\given\xb}^{\perp}}(\gb)(\xb,\thetab)-\zetab(\xb,\thetab)\bigr]}_{\given\xb}^{2}.
        \IEEEyesnumber\label{eqn:pythagoras2}
    \end{IEEEeqnarray*}
    We can then show that 
    $\varPib_{\mathcal{S}_{\varphib\given\xb}}(\gb)(\xb,\thetab) = 
    \Rb_{\gb\varphib\given\xb}\Qb_{\varphib\given\xb}^{-1}\,\varphib(\xb,\thetab)$, and finally, similarly as in Section \ref{ssec:background_cov_ineq}, and after applying $\Expct{\xb}{\cdot}$, we obtain the left inequality in \eqref{eqn:tighter_cov_ineq}.
    
    In addition, as $\varPib_{\mathcal{S}_{\varphib\given\xb}^{\perp}}(\gb)(\xb,\thetab)\in(\mathcal{S}_{\varphib\given\xb}^{\perp})^{L}$ for a.e. $\xb\in\mathcal{S}_{\mathcal{X}}$, $\varPib_{\mathcal{S}_{\varphib\given\cdot}^{\perp}}(\gb)(\cdot)\in(\widetilde{\mathcal{S}}_{\varphib}^{\perp})^{L}$, and since $\widetilde{\mathcal{S}}_{\varphib}^{\perp}\subset\mathcal{S}_{\varphib}^{\perp}$, then $\varPib_{\mathcal{S}_{\varphib\given\cdot}^{\perp}}(\gb)(\cdot)\in(\mathcal{S}_{\varphib}^{\perp})^{L}$. Hence, it is possible to rewrite \eqref{eqn:pythagoras1} with $\zetab(\xb,\thetab)=\varPib_{\mathcal{S}_{\varphib\given\xb}^{\perp}}(\gb)(\xb,\thetab)$, i.e.,
    \begin{IEEEeqnarray*}{'rCl}
        \IEEEeqnarraymulticol{3}{l}{
            \norm[\big]{
                \ab\trsps\bigl[
                    \gb(\xb,\thetab)-\varPib_{\mathcal{S}_{\varphib\given\xb}^{\perp}}(\gb)(\xb,\thetab)
                \bigr]
            }^{2}
        }\\
    \qquad&=&
    \norm[\big]{\ab\trsps\varPib_{\mathcal{S}_{\varphib}}(\gb)(\xb,\thetab)}^{2}
    \\
    &&+\> \norm[\big]{\ab\trsps\bigl[
            \varPib_{\mathcal{S}_{\varphib}^{\perp}}(\gb)(\xb,\thetab)
            - \varPib_{\mathcal{S}_{\varphib\given\xb}^{\perp}}(\gb)(\xb,\thetab)
        \bigr]}^{2},
    \IEEEyesnumber
    \end{IEEEeqnarray*}
    and since $\gb(\xb,\thetab)-\varPib_{\mathcal{S}_{\varphib\given\xb}^{\perp}}(\gb)(\xb,\thetab)=\varPib_{\mathcal{S}_{\varphib\given\xb}}(\gb)(\xb,\thetab)$, we finally obtain
    \begin{equation}
    \norm[\big]{\ab\trsps\varPib_{\mathcal{S}_{\varphib\given\xb}}(\gb)(\xb,\thetab)}^{2}
    \geq
    \norm[\big]{\ab\trsps\varPib_{\mathcal{S}_{\varphib}}(\gb)(\xb,\thetab)}^{2},
    \label{eqn:ineq_BLBs_proj}
    \end{equation}
    that is, for any $\ab\in\Rbb^{L}$,
    \begin{equation}
    \ab\trsps\Expct[\big]{\xb}{
        \Rb_{\gb\varphib\given\xb}
        \Qb_{\varphib\given\xb}^{-1}
        \Rb_{\gb\varphib\given\xb}\trsps
    }\,\ab
    \geq
    \ab\trsps\Rb_{\gb\varphib}
        \Qb_{\varphib}^{-1}
        \Rb_{\gb\varphib}\trsps
    \,\ab,
    \label{eqn:ineq_BLBs_scal}
    \end{equation}
    which proves the right inequality in \eqref{eqn:tighter_cov_ineq}.
\end{IEEEproof}

It is worth noticing that the analytical expressions for the posterior $p(\thetab\given\xb)$ and $p(\xb)$ are not needed for the calculation of the tigther BLBs: from \eqref{eqn:tighter_cov_ineq}, the tighter BLB is given as
\begin{equation}
\mathrm{TBLB} \eqdef \Expct[\big]{\xb}{
    \Rb_{\gb\varphib\given\xb}
    \Qb_{\varphib\given\xb}^{-1}
    \Rb_{\gb\varphib\given\xb}\trsps
}, \label{eqn:TBLB}
\end{equation}
with $\Rb_{\gb\varphib\given\xb}$ and $\Qb_{\varphib\given\xb}$ defined in \eqref{eqn:def_Rgphi_x} and \eqref{eqn:def_Qphi_x} respectively. Using the identity $p(\thetab\given\xb)=p(\xb,\thetab)/p(\xb)$ in \eqref{eqn:def_Rgphi_x} and \eqref{eqn:def_Qphi_x}, we obtain
\begin{IEEEeqnarray}{rCl}
\IEEEyesnumber\IEEEyessubnumber*
\Rb_{\gb\varphib\given\xb} &=&
\frac{\widetilde{\Rb}_{\xb}}{p(\xb)},
\label{eqn:Rgphi_x_bis} \\
\Qb_{\varphib\given\xb}
&=& \frac{\widetilde{\Qb}_{\xb}}{p(\xb)},
\label{eqn:Qphi_x_bis}
\end{IEEEeqnarray}
after defining the integrals
\begin{IEEEeqnarray}{rCl}
\IEEEyesnumber\IEEEyessubnumber*
\widetilde{\Rb}_{\xb}&\eqdef& \int_{\mathcal{S}_{\Theta\given\xb}}\!\!\! \gb(\xb,\thetab)\,\varphib(\xb,\thetab)\trsps\, p(\xb,\thetab) \diff\thetab,
\label{eqn:def_Rtilde_x} \\
\widetilde{\Qb}_{\xb}&\eqdef&\int_{\mathcal{S}_{\Theta\given\xb}}\!\!\! \varphib(\xb,\thetab)\,\varphib(\xb,\thetab)\trsps\, p(\xb,\thetab) \diff\thetab.
\label{eqn:def_Qtilde_x}
\end{IEEEeqnarray}
Substituting \eqref{eqn:Rgphi_x_bis} and \eqref{eqn:Qphi_x_bis} into \eqref{eqn:TBLB}, we obtain
\begin{IEEEeqnarray*}{rCl}
\mathrm{TBLB} &=& \Expct[\big]{\xb}{
    \Rb_{\gb\varphib\given\xb}
    \Qb_{\varphib\given\xb}^{-1}
    \Rb_{\gb\varphib\given\xb}\trsps
} \\
&=& \int_{\mathcal{S}_{\mathcal{X}}}\!\!
    \Rb_{\gb\varphib\given\xb}
    \Qb_{\varphib\given\xb}^{-1}
    \Rb_{\gb\varphib\given\xb}\trsps \, p(\xb)\diff\xb\\
&=& \int_{\mathcal{S}_{\mathcal{X}}}\!\!
    \widetilde{\Rb}_{\xb}
    \widetilde{\Qb}_{\xb}^{-1}
    \widetilde{\Rb}_{\xb}\trsps\diff\xb,
    \IEEEyesnumber
\end{IEEEeqnarray*}
where all $p(\thetab\given\xb)$ and $p(\xb)$ terms have disappeared. Hence, it is not required to know the analytical expressions of these two distributions. We only need to know the expression of $p(\xb,\thetab)$, which is very often readily available, and be able to calculate integrals involving the joint distribution $p(\xb,\thetab)$.

\subsection{The class of Tighter Weiss-Weinstein BLBs} \label{ssec:TWWF}

A tighter counterpart to \eqref{eqn:LB_MSE_WWF} can be obtained by letting $\gb(\xb,\thetab)\eqdef\gb(\thetab)$, and $\zetab(\xb,\thetab)\eqdef\Expct{\thetab|\xb}{\gb(\thetab)}=\widehat{\gb}^{\mathsf{MMSE}}(\xb)$. In order to apply \eqref{eqn:tighter_cov_ineq}, it is necessary that $\widehat{\gb}^{\mathsf{MMSE}}(\xb)\in(\widetilde{\mathcal{S}}_{\varphib}^{\perp})^{L}$, which holds iff, for a.e. $\xb\in\mathcal{S}_{\mathcal{X}}$, and for $\ell=1,\ldots,L$,
\begin{equation*}
\Expct[\big]{\thetab|\xb}{\varphib(\xb,\thetab)\,\widehat{g}_{\ell}^{\mathsf{MMSE}}(\xb)}=
\widehat{g}_{\ell}^{\mathsf{MMSE}}(\xb)\,
\Expct[\big]{\thetab|\xb}{\varphib(\xb,\thetab)}=\bm{0},
\end{equation*}
i.e., iff $\Expct[\big]{\thetab|\xb}{\varphib(\xb,\thetab)}=\bm{0}$ for a.e. $\xb\in\mathcal{S}_{\mathcal{X}}$, that is, iff $\varphib(\cdot)\eqdef\phib(\cdot)\in\mathcal{H}_{\phib}^{L}$ are generating functions of BLBs in the WWF. In such a case, we obtain
\begin{IEEEeqnarray*}{/rCl}
\mathbf{MSE}(\widehat{\gb})&\succeq&\mathbf{MSE}(\widehat{\gb}^{\mathsf{MMSE}}) 
\\
&\succeq& \Expct[\big]{\xb}{\Rb_{\gb\phib\given\xb}\Qb_{\phib\given\xb}^{-1}\,\Rb_{\gb\phib\given\xb}\trsps}
\succeq\Rb_{\gb\phib}\Qb_{\phib}^{-1}\,\Rb_{\gb\phib}\trsps.
\IEEEyesnumber\label{eqn:tighter_LB_MSE_WWF}
\end{IEEEeqnarray*}
As can be seen from \eqref{eqn:ineq_BLBs_proj}, the classical BLB, $\Rb_{\gb\phib}\Qb_{\phib}^{-1}\,\Rb_{\gb\phib}$, relates to squared norms of vectors lying in $\mathcal{S}_{\phib}^{\perp}$ (namely, $\varPib_{\mathcal{S}_{\phib}^{\perp}}(\gb)(\xb,\thetab)$), while the tighter BLB, $\Expct[\big]{\xb}{\Rb_{\gb\phib\given\xb}\Qb_{\phib\given\xb}^{-1}\,\Rb_{\gb\phib\given\xb}\trsps}$, relates to squared norms of vectors lying in $\widetilde{\mathcal{S}}_{\phib}^{\perp}$ (namely, $\varPib_{\mathcal{S}_{\phib|\xb}^{\perp}}(\gb)(\xb,\thetab)$). Yet, $\mathcal{L}_{2}(\mathcal{S}_{\mathcal{X}})\subset\widetilde{\mathcal{S}}_{\phib}^{\perp}\subset\mathcal{S}_{\phib}^{\perp}$, which explains why the tighter BLB indeed is tighter.

\subsection{Condition for equality between classical and tighter BLBs}\label{ssec:equality_BCRBS}

From \eqref{eqn:ineq_BLBs_proj}, the inequality between the classical and the tighter lower bounds in \eqref{eqn:tighter_LB_MSE_WWF} becomes an equality iff 
\begin{equation}
\norm[\big]{\ab\trsps\bigl[
    \varPib_{\mathcal{S}_{\varphib}^{\perp}}(\gb)(\xb,\thetab)
    - \varPib_{\mathcal{S}_{\varphib\given\xb}^{\perp}}(\gb)(\xb,\thetab)
    \bigr]}^{2}=0
\end{equation}
for all $\ab\in\Rbb^{L}$, i.e., iff
\begin{equation} 
    \varPib_{\mathcal{S}_{\varphib}}(\gb)(\xb,\thetab)
    = \varPib_{\mathcal{S}_{\varphib\given\xb}}(\gb)(\xb,\thetab).
    \label{eqn:condition_equality}
\end{equation}
Since $\phib(\cdot)$ is assumed to be made up of $L$ linearly independent functions of $\mathcal{L}_{2}(\mathcal{S}_{\mathcal{X},\Theta})$, the condition for equality \eqref{eqn:condition_equality} is equivalent to
\begin{equation}
    \Rb_{\gb\phib}\Qb_{\phib}^{-1} = \Rb_{\gb\phib\given\xb}\Qb_{\phib\given\xb}^{-1}
    \label{eqn:condition_equality_bis}
\end{equation}
for a.e. $\xb\in\mathcal{S}_{\mathcal{X}}$, which means that $\Rb_{\gb\phib\given\xb}\Qb_{\phib\given\xb}^{-1}$ does actually not depend on $\xb$.

For instance, such a situation occurs in the particular case where $\phib(\xb,\thetab)=\gb(\thetab)-\Expct{\thetab|\xb}{\gb(\thetab)}$, which yields $\Qb_{\phib\given\xb}=\Rb_{\gb\phib\given\xb}$ as well as $\Qb_{\phib}=\Rb_{\gb\phib}=\mathbf{MSE}(\widehat{\gb}^{\mathsf{MMSE}})$. Then, \eqref{eqn:tighter_LB_MSE_WWF} simply reduces to
\begin{equation}
    \mathbf{MSE}(\widehat{\gb})\succeq\mathbf{MSE}(\widehat{\gb}^{\mathsf{MMSE}}).
\end{equation}
As shown in the following, the choice $\phib(\xb,\thetab)=\gb(\thetab)-\Expct{\thetab|\xb}{\gb(\thetab)}$ is not the only one that leads to \eqref{eqn:condition_equality_bis}.

\section{Tighter Bayesian Cramér-Rao bounds and an update for the notion of efficiency} \label{sec:TBCRB}

\subsection{Classical and tighter Bayesian Cramér-Rao bounds}\label{ssec:BCRBs}

Let us consider the family of $M=K$ functions
$\phib^{\mathrm{CR}}(\cdot)=
(\phi_{1}^{\mathrm{CR}}(\cdot),\ldots,\phi_{K}^{\mathrm{CR}}(\cdot))\trsps$ defined as
\begin{equation}
    \phib^{\mathrm{CR}}(\xb,\thetab)\eqdef
    \left\{
    \begin{IEEEeqnarraybox}[
        \IEEEeqnarraystrutmode
        \IEEEeqnarraystrutsizeadd{7pt}{0pt}
    ][c]{c/l}
    \frac{\partial \ln p(\thetab\given\xb)}{\partial\thetab}, &\text{if } (\xb,\thetab)\in\mathcal{S}_{\mathcal{X},\Theta}, \\
    0, & \text{otherwise}.
    \end{IEEEeqnarraybox}
    \right.
\label{eqn:def_phiCR}
\end{equation}
Sufficient conditions for $\phib^{\mathrm{CR}}(\cdot)$ to generate BLBs of the WWF are \cite{WW88}
\begin{IEEEitemize}
    \item $\mathcal{S}_{\Theta\given\xb}=\mathcal{I}_{\xb,1}\times\ldots\times\mathcal{I}_{\xb,K}$, where $\mathcal{I}_{\xb,k}=\intervaloo{a_{\xb,k}}{b_{\xb,k}}$ are intervals of $\Rbb$ with endpoints $a_{\xb,k},b_{\xb,k}\in\intervalcc{-\infty}{+\infty}$, $a_{\xb,k}<b_{\xb,k}$, $k=1,\ldots,K$;
    
    \item $p(\thetab\given\xb)$ is absolutely continuous w.r.t. $\theta_{k}$, $k=1,\ldots,K$, for a.e. $\xb\in\mathcal{S}_{\mathcal{X}}$;
    
    \item $\lim_{\theta_{k}\rightarrow a_{\xb,k}} g_{\ell}(\thetab)\,p(\thetab\given\xb)=\lim_{\theta_{k}\rightarrow b_{\xb,k}} g_{\ell}(\thetab)\,p(\thetab\given\xb)=0$, $\ell=1,\ldots,L$, $k=1,\ldots,K$, for a.e. $\xb\in\mathcal{S}_{\mathcal{X}}$;
    
    \item The (joint) Bayesian information matrix $\Qb_{\phib^{\mathrm{CR}}}$ is non-singular.
\end{IEEEitemize}
After plugging \eqref{eqn:def_phiCR} into \eqref{eqn:def_Rgphi} and \eqref{eqn:def_Qphi}, we obtain the classical version of the Bayesian Cramér-Rao bound (BCRB), that is
\begin{equation}
    \mathbf{BCRB}=
    \Rb_{\gb\phib^{\mathrm{CR}}}
    \Qb_{\phib^{\mathrm{CR}}}^{-1}
    \Rb_{\gb\phib^{\mathrm{CR}}}\trsps,
    \label{eqn:BCRB_general}
\end{equation}
with
\begin{IEEEeqnarray}{;l} 
    \Rb_{\gb\phib^{\mathrm{CR}}}=
    \Expct*{\xb,\thetab}{
        \gb(\thetab)\frac{\partial \ln p(\thetab\given\xb)}{\partial\thetab\trsps}
    }
    =-\Expct*{\xb,\thetab}{
        \frac{\partial\gb(\thetab)}{\partial\thetab\trsps}
    },
    \label{eqn:Rgphi_CR}
\end{IEEEeqnarray}
and
\begin{equation}
    \Qb_{\phib^{\mathrm{CR}}}=
    \Expct*{\xb,\thetab}{
        \frac{\partial \ln p(\thetab\given\xb)}{\partial\thetab}
        \frac{\partial \ln p(\thetab\given\xb)}{\partial\thetab\trsps}
    }.
\label{eqn:Qphi_CR}
\end{equation}
Similarly, the tighter version of the BCRB is obtained by plugging \eqref{eqn:def_phiCR} into \eqref{eqn:def_Rgphi_x} and \eqref{eqn:def_Qphi_x}, which gives
\begin{equation}
    \mathbf{TBCRB}=
    \Expct[\big]{\xb}{
        \Rb_{\gb\phib^{\mathrm{CR}}\given\xb}
        \Qb_{\phib^{\mathrm{CR}}\given\xb}^{-1}
        \Rb_{\gb\phib^{\mathrm{CR}}\given\xb}\trsps
    },
\label{eqn:TBCRB_general}
\end{equation}
with
\begin{IEEEeqnarray}{-l}
    \Rb_{\gb\phib^{\mathrm{CR}}\given\xb}=
    \Expct*{\thetab\given\xb}{
        \gb(\thetab)\frac{\partial \ln p(\thetab\given\xb)}{\partial\thetab\trsps}
    }
    =-\Expct*{\thetab\given\xb}{
        \frac{\partial\gb(\thetab)}{\partial\thetab\trsps}
    }\!,
    \label{eqn:Rgphi_x_CR}
\end{IEEEeqnarray}
and
\begin{equation}
    \Qb_{\phib^{\mathrm{CR}}\given\xb}=
    \Expct*{\thetab\given\xb}{
        \frac{\partial \ln p(\thetab\given\xb)}{\partial\thetab}
        \frac{\partial \ln p(\thetab\given\xb)}{\partial\thetab\trsps}
    },
\label{eqn:Qphi_x_CR}
\end{equation}
that can be referred to as the posterior (or conditional) Bayesian information matrix. Finally, \eqref{eqn:tighter_LB_MSE_WWF} can be rewritten as
\begin{equation}
    \mathbf{MSE}(\widehat{\gb})\succeq\mathbf{MSE}(\widehat{\gb}^{\mathsf{MMSE}})\succeq\mathbf{TBCRB}\succeq\mathbf{BCRB}.
    \label{eqn:LB_MSE_BCRBs}
\end{equation}

We can remark that an alternative way of obtaining these two versions of the BCRB consists in starting from the functions $\phi_{\hb_{1}}^{\mathrm{BZ}}(\cdot),\ldots,\phi_{\hb_{M}}^{\mathrm{BZ}}(\cdot)$ generating the Bobrovsky-Zakaï bound \cite{BZ76}, which are defined, for $\Hb=[\hb_{1},\ldots,\hb_{M}]\in\Rbb^{K\times M}$ and $m=1,\ldots,M$, as
\begin{IEEEeqnarray*}{-l}
\phi_{\hb_{m}}^{\mathrm{BZ}}(\xb,\thetab)= \\
\quad \frac{
    p(\thetab+\hb_{m}\given\xb)\,\indicf{\mathcal{S}_{\Theta\given\xb}}{\thetab+\hb_{m}} 
    -p(\thetab\given\xb)\,\indicf{\mathcal{S}_{\Theta\given\xb}}{\thetab-\hb_{m}}
}{
    p(\thetab\given\xb)\,\indicf{\mathcal{S}_{\Theta\given\xb}}{\thetab}
},\\[-1.5ex]
\IEEEyesnumber\label{eqn:def_phiBZ}
\end{IEEEeqnarray*}%
if $(\xb,\thetab) \in \mathcal{S}_{\mathcal{X},\Theta }$, and $\phi_{\hb_{m}}^{\mathrm{BZ}}(\xb,\thetab) =0$ otherwise, and by letting $\hb_{m}$ tend to $\bm{0}$, $m=1,\ldots,M$. In particular, this allows for deriving conditions under which the BCRB, and thus the TBCRB, are nonzero (see \cite{CRE17,CF18} for instance).

\subsection{An update for the notion of efficiency}

Since the two BCRBs in \eqref{eqn:LB_MSE_BCRBs} are nonzero under the same conditions (see \cite[Section~III]{CRE17}), it seems appropriate to define the notion of efficiency as follows.
\begin{definition}\label{def:efficiency}
    An estimator $\widehat{\gb}(\cdot)$ of $\gb(\thetab)$ is said to be efficient if its MSE achieves the TBCRB. In other terms, $\mathbf{MSE}(\widehat{\gb})=\mathbf{TBCRB}$ in \eqref{eqn:LB_MSE_BCRBs}.
\end{definition}

Indeed, it has often been noticed that the classical BCRB cannot be attained for plenty of estimation problems. From \eqref{eqn:LB_MSE_BCRBs} and Definition \ref{def:efficiency}, it can be seen that if an estimator, different from the MMSE estimator, is efficient, then so is the MMSE estimator. Let us now investigate conditions for the TBCRB to be attained, hence allowing for a description of models for which efficient estimators (potentially different from the MMSE estimator) could be found. 

\subsection{Class of efficient estimators}

Rewriting \eqref{eqn:pythagoras2} with $\gb(\xb,\thetab)\eqdef\gb(\thetab)$ and $\zetab(\xb,\thetab)\eqdef\widehat{\gb}(\xb)$, 
we have, for any $\ab\in\Rbb^{L}$, $\phib(\cdot)\in\mathcal{H}_{\phib}$, and a.e. $\xb\in\mathcal{X}$,
\begin{IEEEeqnarray*}{-rCl}
    \norm{
        \ab\trsps(\gb(\thetab)-\widehat{\gb}(\xb))
    }_{|\xb}^{2}
    &=&\norm{
        \ab\trsps\varPib_{\mathcal{S}_{\phib|\xb}}(\gb)(\xb,\thetab)
    }_{|\xb}^{2}\\
    &&
    \hspace{-2.5em}
    {+}\>\norm[\big]{
        \ab\trsps\bigl[
            \varPib_{\mathcal{S}_{\phib|\xb}^{\perp}}(\gb)(\xb,\thetab)
            -\widehat{\gb}(\xb)
        \bigr]
    }_{|\xb}^{2}.
\IEEEyesnumber\label{eqn:pythagoras2bis}
\end{IEEEeqnarray*}
Consequently, the estimator $\widehat{\gb}(\cdot)$ is efficient iff
$\norm{
    \ab\trsps(\gb(\thetab)-\widehat{\gb}(\xb))
}_{|\xb}^{2}
=\norm{
    \ab\trsps\varPib_{\mathcal{S}_{\phib|\xb}}(\gb)(\xb,\thetab)
}_{|\xb}^{2}$ for any $\ab\in\Rbb^{L}$ and a.e. $\xb\in\mathcal{X}$, that is, from \eqref{eqn:pythagoras2bis}, iff
\begin{equation}
    \norm[\big]{
        \ab\trsps\bigl[
        \varPib_{\mathcal{S}_{\phib|\xb}^{\perp}}(\gb)(\xb,\thetab)
        -\widehat{\gb}(\xb)
        \bigr]
    }_{|\xb}^{2}=0
\end{equation}
for all $\ab\in\Rbb^{L}$, i.e., iff
\begin{equation}
    \varPib_{\mathcal{S}_{\phib|\xb}^{\perp}}(\gb)(\xb,\thetab)
    =\widehat{\gb}(\xb).
    \label{eqn:condition_efficiency}
\end{equation}
Since $\gb(\thetab)=\varPib_{\mathcal{S}_{\phib|\xb}}(\gb)(\xb,\thetab)+\varPib_{\mathcal{S}_{\phib|\xb}^{\perp}}(\gb)(\xb,\thetab)$, \eqref{eqn:condition_efficiency} is finally equivalent to
\begin{equation}
    \gb(\thetab)-\widehat{\gb}(\xb)=\varPib_{\mathcal{S}_{\phib|\xb}}(\gb)(\xb,\thetab)= \Rb_{\gb\phib|\xb}\Qb_{\phib|\xb}^{-1}\phib(\xb,\thetab).
    \label{eqn:condition_efficiency_bis}
\end{equation}

In the case of the TBCRB, i.e., $\phib(\cdot)=\phib^{\mathrm{CR}}(\cdot)$, the condition \eqref{eqn:condition_efficiency_bis} becomes, from \eqref{eqn:def_Rgphi_x} and \eqref{eqn:def_Qphi_x},
\begin{IEEEeqnarray*}{-l}
    \widehat{\gb}(\xb)-\gb(\thetab)= \\
    \Expct*{\thetab\given\xb}{
        \frac{\partial\gb(\thetab)}{\partial\thetab\trsps}
    }\!
    \Expct*{\thetab\given\xb}{
        \frac{\partial \ln p(\thetab\given\xb)}{\partial\thetab}
        \frac{\partial \ln p(\thetab\given\xb)}{\partial\thetab\trsps}
    }^{-1}\!
    \frac{\partial \ln p(\thetab\given\xb)}{\partial\thetab}.\\[-1ex]
    \IEEEyesnumber\label{eqn:condition_efficiency_ter}
\end{IEEEeqnarray*}
From this condition, a number of particular cases can be examined.

\subsubsection{Scalar case ($K=L=1$)}\label{sssec:scalar_case}

In the scalar case, the function $\gb(\thetab)\eqdef g(\theta)$ reduces to a single function of a scalar parameter $\theta\in\Theta\subset\Rbb$. Then, \eqref{eqn:condition_efficiency_ter} reduces to
\begin{equation}
    \widehat{g}(\xb)-g(\theta)=
    v(\xb)
    \frac{\partial \ln p(\theta\given\xb)}{\partial\theta},
    \label{eqn:diff_eq}
\end{equation}
where
\begin{equation}
    v(\xb) = \frac{
        \Expct{\theta|\xb}{
            \frac{\diff g(\theta)}{\diff\theta}
        }
    }{
        \Expct*{\theta|\xb}{
            \bigl(
            \frac{\partial \ln p(\theta\given\xb)}{\partial\theta}
            \bigr)^{2}
        }
    }.
    \label{eqn:expr_v}
\end{equation}
Integrating \eqref{eqn:diff_eq} w.r.t. $\theta$ leads to
\begin{equation}
    p(\theta\given\xb) = b(\xb)\exp\biggl[
        \frac{\theta\widehat{g}(\xb)-G(\theta)}{v(\xb)}
    \biggr],
    \label{eqn:efficient_posterior}
\end{equation}
in which
\begin{equation}
    b(\xb) =\biggl(
        \int_{\mathcal{S}_{\Theta|\xb}}
        \exp{\left[
        \frac{\theta\widehat{g}(\xb)-G(\theta)}{v(\xb)}\right]}\diff\theta
    \biggr)^{-1},
\end{equation}
and $G(\theta)$ is the antiderivative of $g(\theta)$, i.e.,
$\diff G(\theta)/\diff\theta = g(\theta)$. It is worth noticing that the posterior distribution $p(\theta\given\xb)$ in \eqref{eqn:efficient_posterior} is of the form
\begin{equation}
    f_{\theta}(\xb) = h(\xb)\exp\bigl[
        \etab(\theta)\trsps \tb(\xb) - A(\theta)
    \bigr]
    \label{eqn:exponential_family}
\end{equation}
with $h(\xb)=b(\xb)$, $\etab(\theta)=(\theta,G(\theta))\trsps$, 
$\tb(\xb)=(\widehat{g}(\xb)/v(\xb),-1/v(\xb))\trsps$, and $A(\theta)=0$, meaning that $p(\theta\given\xb)$ is an exponential family, which turns out to be curved, due to the nonlinear link between the components of $\tb(\xb)$ \cite[pp.~23--32]{LC03}. Moreover, the form of the posterior distribution \eqref{eqn:efficient_posterior} leads to the following form of the joint distribution $p(\xb,\theta)$:
\begin{equation}
     p(\xb,\theta) = b(\xb)\,p(\xb)\,\exp\biggl[
    \frac{\theta\widehat{g}(\xb)-G(\theta)}{v(\xb)}
    \biggr].
    \label{eqn:joint_exp}
\end{equation}
Hence, the class of models for which an efficient estimator can be found are those in which the joint distribution can be written in the form \eqref{eqn:joint_exp}.

\subsubsection{Exponential family models}

In the scalar case, studied in the previous section, 
numerous statistical models are concerned with a likelihood function that belongs to a general exponential family \eqref{eqn:exponential_family}, i.e., there exists functions $h(\xb)\geq 0$, $\etab(\theta),\tb(\xb)\in\Rbb^{J}$ and $A(\theta)$ such that
\begin{equation}
p(\xb\given\theta) = h(\xb)\exp\bigl[
\etab(\theta)\trsps \tb(\xb) - A(\theta)
\bigr],
\label{eqn:exponential_L}
\end{equation}
with 
$A(\theta)=\ln
\int_{\mathcal{S}_{\mathcal{X}|\theta}}
h(\xb)\exp\bigl[\etab(\theta)\trsps \tb(\xb)\bigr]
\diff\xb$. In the Bayesian framework, assigning a conjugate distribution to the prior w.r.t. the likelihood function ensures that the posterior p.d.f. is in the same probability distribution family, and often makes it more easily tractable. For a likelihood function of the form in \eqref{eqn:exponential_L}, the conjugate prior has the form
\begin{equation}
p(\theta)=p(\theta\semcol\lambda,\mub)=
\kappa(\lambda,\mub)\exp\bigl[
\etab(\theta)\trsps\mub-\lambda A(\theta)
\bigr],
\label{eqn:exponential_prior}
\end{equation}
with $\kappa(\lambda,\mub)=1/\int_{\mathcal{S}_{\Theta}}\exp[
\etab(\theta)\trsps\mub-\lambda A(\theta)
]\diff\theta$, i.e., the form of an exponential family as well. This leads to the following joint p.d.f.:
\begin{IEEEeqnarray*}{rCl}
p(\xb,\theta)&=&p(\xb,\theta\semcol\lambda,\mub) \\
&=&p(\xb\given\theta)\,p(\theta\semcol\lambda,\mub)\\
&=&\kappa(\lambda,\mub)\,h(\xb)\\
& &\times\exp\bigl[
\etab(\theta)\trsps(\mub+\tb(\xb))
-(\lambda+1)A(\theta)
\bigr],
\IEEEyesnumber\label{eqn:exponential_joint}%
\end{IEEEeqnarray*}
and to the following marginal distribution of the observations:
\begin{IEEEeqnarray*}{'rCl}
p(\xb)&=&p(\xb\semcol\lambda,\mub)
=\int_{\mathcal{S}_{\Theta\given\xb}}\!\!\!p(\xb,\theta\semcol\lambda,\mub)\diff\theta \\
&=& \kappa(\lambda,\mub)\,h(\xb)\\
&&\times\int_{\mathcal{S}_{\Theta\given\xb}}\hspace{-1em}
\exp\bigl[
\etab(\theta)\trsps(\mub+\tb(\xb))-(\lambda+1)A(\theta)
\bigr]
\diff\theta\\
&=&\frac{\kappa(\lambda,\mub)\,h(\xb)}{\kappa(\lambda+1,\mub+\tb(\xb))}.
\IEEEyesnumber\label{eqn:exponential_marginal}
\end{IEEEeqnarray*}
Finally, the \textit{a posteriori} distribution can be obtained as
\begin{IEEEeqnarray*}{-rCl}
p(\theta\given\xb\semcol\lambda,\mub)
&=& \frac{p(\xb,\theta\semcol\lambda,\mub)}{p(\xb\semcol\lambda,\mub)}\\
&=& \kappa(\lambda+1,\mub+\tb(\xb))\\
&& \times\exp\bigl[
\etab(\theta)\trsps(\mub+\tb(\xb))
-(\lambda+1)A(\theta)
\bigr]\!,
\IEEEyesnumber\label{eqn:exponential_posterior}%
\end{IEEEeqnarray*}
which has, as expected, the same form as the prior distribution \eqref{eqn:exponential_prior}, i.e.,
$
p(\theta\given\xb\semcol\lambda,\mub) = 
p(\theta\semcol\lambda+1,\mub+\tb(\xb)).
$
Yet, for the TBCRB to be attained, it is necessary that \eqref{eqn:exponential_posterior} has the form \eqref{eqn:efficient_posterior}, i.e., iff
\begin{equation}
\etab(\theta)\trsps(\mub+\tb(\xb))-(\lambda+1)A(\theta)
= \frac{\theta\widehat{g}(\xb)-G(\theta)}{v(\xb)},
\label{eqn:condition_efficiency_exponential}
\end{equation}
which automatically implies $\kappa(\lambda+1\semcol\mub+\tb(\xb))=b(\xb)$. Condition for efficiency \eqref{eqn:condition_efficiency_exponential} may allow for derivation of efficient estimators. For instance, let us consider the particular case where $v(\xb)$ does actually not depend upon $\xb$. Without loss of generality, let us assume $v(\xb)=1$. Then, \eqref{eqn:condition_efficiency_exponential} can be rewritten as
\begin{equation}
G(\theta)+\etab(\theta)\trsps\mub-(\lambda+1)A(\theta)
=\theta\widehat{g}(\xb)-\etab(\theta)\trsps\tb(\xb).
\label{eqn:condition_efficiency_exponential_bis}
\end{equation}
By differentiating both sides w.r.t. $\xb$, it appears that $\etab(\theta)=\cb\theta$, where $\cb$ is a constant vector, i.e., $\theta$ is a natural parameter for the likelihood function \eqref{eqn:exponential_L} and the prior p.d.f. \eqref{eqn:exponential_prior}. After differentiating both sides of \eqref{eqn:condition_efficiency_exponential_bis} w.r.t. $\theta$, we can deduce that 
\begin{subnumcases}{}
g(\theta) = (\lambda+1)\frac{\diff A(\theta)}{\diff\theta}+C,
\label{eqn:suff_condition_efficiency_g}\\
\widehat{g}(\xb) = \cb\trsps\tb(\xb)+C'
\label{eqn:suff_condition_efficiency_ghat}
\end{subnumcases}
where $C$ and $C'$ denote some scalar constants. Finally, in the case of signals modeled by \eqref{eqn:exponential_L} and \eqref{eqn:exponential_prior}, the estimator given by \eqref{eqn:suff_condition_efficiency_ghat} provides an efficient estimate of $g(\theta)$ given by \eqref{eqn:suff_condition_efficiency_g}, with an associated MSE equal to 1. Since $v(\xb)$ has been assumed constant w.r.t. $\xb$, it can be noticed that the related TBCRB and BCRB are equal. However, as shown in the next section, efficient estimation is also possible when $v(\xb)$ does depend upon $\xb$, and in that case the TBCRB differs from the BCRB (see Section~\ref{ssec:equality_BCRBS}).

\subsubsection{Particular case: $g(\theta)=\theta$}

A case of particular interest is that where $g(\theta)\eqdef\theta$. Then, \eqref{eqn:expr_v} reduces to
\begin{equation}
    v(\xb)=\frac{1}{\Expct*{\theta|\xb}{
            \bigl(
                \frac{\partial\ln p(\theta\given\xb)}{\partial\theta}
            \bigr)^{2}}
        },
    \label{eqn:posterior_variance}
\end{equation}
and \eqref{eqn:efficient_posterior} becomes
\begin{IEEEeqnarray*}{rCl}
    p(\theta\given\xb)&=&b(\xb)\exp\biggl[
        \frac{
            \theta\widehat{\theta}(\xb)-\frac{\theta^{2}}{2}+c
        }{
            v(\xb)
        }
    \biggr]\\
    &=& b(\xb)\exp\biggl[
        \frac{2c+\widehat{\theta}(\xb)^{2}}{2v(\xb)}
    \biggr]
    \exp\biggl[
        -\frac{(\theta-\widehat{\theta}(\xb))^{2}}{2v(\xb)}
    \biggr]\\
    &=&\frac{1}{\sqrt{2\pi v(\xb)}}\exp\biggl[
    -\frac{(\theta-\widehat{\theta}(\xb))^{2}}{2v(\xb)}
    \biggr], \IEEEyesnumber\label{eqn:posterior_gauss}
\end{IEEEeqnarray*}
where $c$ denotes a scalar constant, and the last line is obtained from the condition 
$\int_{\mathcal{S}_{\Theta|\xb}} p(\theta\given\xb)\diff\theta=1$.
Consequently, it is possible to find an efficient estimate of $\theta$ only if the posterior p.d.f. $p(\theta\given\xb)$ is Gaussian for a.e. $\xb\in\mathcal{X}$. It is important to notice that, in that case, the posterior variance $\sigma_{|\xb}^{2}\eqdef v(\xb)$ may depend upon $\xb$. This result differs from and relaxes Van Trees' as for the attainment of the standard BCRB \cite[p.73]{VT68}, where the posterior variance cannot depend upon $\xb$, since the conditional expectation $\Expct{\theta|\xb}{\cdot}$ in \eqref{eqn:posterior_variance} is replaced with the joint expectation $\Expct{\xb,\theta}{\cdot}$. This remark is in agreement with that made in Section \ref{ssec:equality_BCRBS}, that is, if the posterior variance $\sigma_{|\xb}^{2}$ does depend upon $\xb$, then $\mathrm{TBCRB}>\mathrm{BCRB}$. However, it can also be noted that, if the posterior p.d.f. is not Gaussian, then $\mathrm{MSE}(\widehat{\theta})>\mathrm{TBCRB}\geq\mathrm{BCRB}$, i.e., neither the TBCRB, nor the BCRB, are tight BLBs.

In order to illustrate these points, we study a practical example in the next section.

\section{Study of a noteworthy example} \label{sec:example}

Let us consider the same problem as in \cite[p.7]{VTB07}, of estimating the variance of a Gaussian random variable with known mean (or equivalently assumed zero), from an observation vector $\xb$ consisting of $N$ i.i.d. samples $x_{1},\ldots,x_{N}$, i.e., $\xb\sim\mathcal{N}(0,\sigma^{2}\Ib)$. The parameter of interest is $\theta\eqdef\sigma^{2}$, so the likelihood function can be written as
\begin{equation}
p(\xb\given\theta) = (2\pi)^{-\frac{N}{2}}\,\theta^{-\frac{N}{2}}\,e^{-\frac{\xb\trsps\xb}{2\theta}}.
\label{eqn:L}
\end{equation}
We assume $\theta$ a priori follows a beta distribution, i.e., for $0<\theta\leq 1$,
\begin{equation}
p(\theta\semcol a)=B(a,a)^{-1}\,\theta^{a-1}(1-\theta)^{a-1},
\label{eqn:prior}
\end{equation}
where $B(a,a)\eqdef\int_{0}^{1}\theta^{a}\,(1-\theta)^{1-a}\diff\theta=\varGamma^{2}(a)/\varGamma(2a)$, and $\varGamma(\cdot)$ denotes the gamma function: $\varGamma(a)\eqdef\int_{0}^{\infty}t^{a-1}\,e^{-t}\diff t$. This prior distribution is symmetric, with mean $\mu_{\pi}=1/2$ and variance $\sigma_{\pi}^{2}=1/4(2a+1)$. For $a=1$, it reduces to a uniform distribution on $\intervalcc{0}{1}$. As $a$ increases, it becomes narrower, and when $a\rightarrow+\infty$, we approach the case of $\theta$ known. Multiplying \eqref{eqn:L} by \eqref{eqn:prior} yields the joint distribution, for $\xb\in\Rbb^{N}$ and $0\leq\theta\leq 1$, as
\begin{equation}
p(\xb,\theta\semcol a)=(2\pi)^{-\frac{N}{2}}\,B(a,a)^{-1}\,\theta^{a-\frac{N}{2}-1}(1-\theta)^{a-1}\,e^{-\frac{\xb\trsps\xb}{2\theta}}.
\label{eqn:joint_pdf}
\end{equation}
In the following, let us recall some known results on the BCRB, the expected CRB (ECRB), and the maximum a posteriori (MAP) estimator.

\subsection{BCRB, MAP estimator and ECRB (known results)} \label{ssec:exple_known_results}

Since $\partial\ln p(\theta\given\xb)/\partial\theta=\partial\ln p(\xb,\theta)/\partial\theta$, the Bayesian Fisher information can be obtained from the joint p.d.f. \eqref{eqn:joint_pdf}~as
\begin{IEEEeqnarray*}{rCl}
    F_{\mathrm{B}}&=&\Expct*{\xb,\theta}{\Bigl(
            \frac{\partial\ln p(\xb,\theta)}{\partial\theta}
        \Bigr)^{2}}
    =-\Expct*{\xb,\theta}{
            \frac{\partial^{2}\ln p(\xb,\theta)}{\partial\theta^{2}}
        } \\
    &=& \Expct[\bigg]{\xb,\theta}{
            \frac{a-1-\frac{N}{2}}{\theta^{2}}+\frac{\xb\trsps\xb}{\theta^{3}}+\frac{a-1}{(1-\theta)^{2}}
        }. \IEEEyesnumber\label{eqn:def_BFIM}
\end{IEEEeqnarray*}
After computing the expectation, we obtain
\begin{equation}
    F_{\mathrm{B}}=(N+4(a-1))
      \frac{\varGamma(a-2)\,\varGamma(2a)}{2\,\varGamma(2a-2)\,\varGamma(a)},
      \label{eqn:BFIM}
\end{equation}
which reduces, for $a>2$, to
\begin{equation}
    F_{\mathrm{B}}=\frac{(N+4(a-1))(2a-1)}{a-2},
    \label{eqn:BFIM_asup2}
\end{equation}
and the BCRB is simply given by
\begin{equation}
    \mathrm{BCRB}=\frac{1}{F_{\mathrm{B}}}.
    \label{eqn:BCRB}
\end{equation}

The MAP estimator, defined as
\begin{equation}
    \widehat{\theta}^{\mathsf{MAP}}(\xb)\eqdef\argmax_{0\leq\theta\leq 1} p(\theta\given\xb)=\argmax_{0\leq\theta\leq 1} p(\xb,\theta),
\end{equation}
is shown to be equal to \cite[p.9]{VTB07}
\begin{subnumcases}{\hspace{-1em}\widehat{\theta}^{\mathsf{MAP}}(\xb)=}
    \frac{\beta-\sqrt{\beta^{2}-4\alpha\gamma}}{2\alpha},
    &\hspace{-1em}if $N\neq 4(a-1)$, 
    \label{eqn:MAPE_N<>4(a-1)}\\
    \frac{\gamma}{\frac{1}{2}+\gamma},
    &\hspace{-1em}if $N=4(a-1)$,
    \label{eqn:MAPE_N=4(a-1)}
\end{subnumcases}
where $\alpha=1-4(a-1)/N$, $\beta=1-2(a-1)/N+\gamma$, and $\gamma=\xb\trsps\xb/N$.

An approximation for the asymptotic MSE of $\widehat{\theta}^{\mathsf{MAP}}$ can be obtained by computing the ECRB, defined by
\begin{equation}
    \mathrm{ECRB}=\Expct{\theta}{\mathrm{CRB}(\theta)}
    =\Expct{\theta}{F^{-1}(\theta)},
\end{equation}
with $\mathrm{CRB}(\theta)$ denoting the classical Cramér-Rao bound, that is the inverse of the classical Fisher information for the parameter $\theta$:
\begin{equation}
    \mathrm{CRB}(\theta)=\frac{1}{F(\theta)}
    =-\frac{1}{\Expct[\big]{\xb|\theta}{
            \frac{\partial^{2}\ln p(\xb\given\theta)}{\partial\theta^{2}}
        }}.
\end{equation}
For the example under study, the ECRB is given by \cite[p.11]{VTB07}
\begin{equation}
    \mathrm{ECRB}=\frac{2}{N}\Expct{\theta}{\theta^{2}}
    =\frac{1}{N}\frac{a+1}{2a+1}.
    \label{eqn:ECRB_exple}
\end{equation}
It should be noted that, even if the ECRB approximates the MSE of the MAP estimator in the asymptotic regime, it is \emph{not} a BLB on the MSE, and thus cannot provide any insight about estimation performance in other regimes.

\subsection{Expressions of $p(\xb)$, $p(\theta\given\xb)$, $\mathrm{E}_{\theta|\xb}[\theta]$ and MMSE}

In order to proceed with further results, the following relation is essential: according to \cite[3.471(2.)]{GR94},
$\forall \lambda,\mu,\nu \in\Cbb,\mathrm{Re}(\lambda) >0, \mathrm{Re}(\mu) >0$,
\begin{equation}
    \int_{0}^{1}\!\!\theta^{\nu-1}(1-\theta) ^{\mu -1}e^{-%
        \frac{\lambda }{\theta }}\diff\theta =\lambda ^{\frac{\nu -1}{2}}e^{-\frac{\lambda }{2}}\varGamma(\mu )\, W_{\frac{1-2\mu -\nu }{2},\frac{\nu }{2}}(
    \lambda) ,  \label{eqn:GR_identity}
\end{equation}%
where $W_{\mu ,\nu }\left( \cdot\right) $ denotes a Whittaker function \cite[\S9.22--9.23]{GR94}. By setting $\xi=\frac{N-6a+2}{4}$, and applying (\ref{eqn:GR_identity}) to
\begin{equation}
    p(\xb) =\left( 2\pi \right) ^{-\frac{N}{2}}B(a,a) ^{-1}\!\!\int_{0}^{1}\!\!\theta ^{a-\frac{N}{2}-1}(1-\theta) ^{a-1}e^{-\frac{1}{\theta }\frac{\xb\trsps%
            \xb}{2}}\diff\theta ,
\end{equation}%
we obtain the following closed-form expressions for the marginal and the posterior distributions:
\begin{IEEEeqnarray*}{-rCl}
    p\left( \xb\right) &=& \frac{\varGamma(a) \bigl(\frac{\xb\trsps\xb}{2}%
        \bigr) ^{a-\frac{N}{2}-1}e^{-\frac{\xb\trsps\xb}{4}}\,
        W_{\xi,\frac{2a-N}{4}}\bigl( \frac{\xb\trsps\xb}{2}\bigr)}{%
        \sqrt{2\pi }^{N}B(a,a) }, 
    \IEEEyesnumber\IEEEyessubnumber*\label{eqn:marginal_pdf}
    \\
    p\left( \theta \given\xb\right) &=&
    \frac{\theta ^{a-\frac{N}{2}-1}(1-\theta)^{a-1}\,e^{-\frac{1}{%
                \theta }\frac{\xb\trsps\xb}{2}}}{\varGamma(a)\,\bigl(
        \frac{\xb\trsps\xb}{2}\bigr) ^{\frac{2a-N-2}{4}}e^{-%
            \frac{\xb\trsps\xb}{4}}\,W_{\xi,\frac{2a-%
                N}{4}}\bigl( \frac{\xb\trsps\xb}{2}\bigr) }. 
    \label{eqn:posterior_pdf}
\end{IEEEeqnarray*}%
Similarly, the MMSE estimator can be obtained as
\begin{IEEEeqnarray*}{rCl}
    \widehat{\theta}^{\mathsf{MMSE}}(\xb)&=&\Expct{\theta|\xb}{\theta}
    =\int_{0}^{1}\!\!\theta\,p( \theta \given\xb) \diff\theta 
    =\frac{\int_{0}^{1}\theta\, p( \xb,\theta ) \diff\theta }{p(\xb)} \\
    &=&\frac{
            \int_{0}^{1}
            \theta^{a-\frac{N}{2}}
            (1-\theta)^{a-1}\,
            e^{-\frac{1}{\theta}\frac{\xb\trsps\xb}{2}}\diff\theta
        }{
            p(\xb)\sqrt{2\pi}^{N}B(a,a)
        }\\
    &=&\sqrt{\frac{\xb\trsps\xb}{2}}\,
       \frac{
          W_{\xi-\frac{1}{2},\frac{2a-N+2}{4}}\bigl( \frac{\xb\trsps\xb}{2}\bigr)
       }{
          W_{\xi,\frac{2a-N}{4}}\bigl( \frac{\xb\trsps\xb}{2}\bigr) 
       },
    \IEEEyesnumber \label{eqn:MMSEE}
\end{IEEEeqnarray*}%
as well as the MMSE itself, as
\begin{equation}
    \mathrm{MMSE}=\Expct[\big]{\xb,\theta }{
        (\Expct{\theta|\xb}{\theta} -\theta)^{2}
    }
    =\Expct*{\xb}{
        \Expct{\theta|\xb}{\theta^{2}}-\Expct{\theta|\xb}{\theta}^{2}
    }, \label{eqn:MMSE_v0}
\end{equation}%
where, using the same type of derivation as for \eqref{eqn:MMSEE}, we have
\begin{equation}
    \Expct{\theta|\xb}{\theta^{2}} =
    \frac{\int_{0}^{1}\!\theta^{2}\, p( \xb,\theta ) \diff\theta}{p(\xb)}
    = \frac{\xb\trsps\xb}{2}\frac{W_{\xi-1,%
            \frac{2a-N}{4}+1}\bigl( \frac{\xb\trsps\xb}{2}\bigr) }{%
        W_{\xi,\frac{2a-N}{4}}\bigl( \frac{\xb%
            \trsps\xb}{2}\bigr) }, \label{eqn:E_theta2}
\end{equation}%
hence giving
\begin{IEEEeqnarray*}{'l}
    \mathrm{MMSE}=
    \ExpctE_{\xb}\Biggl[ \frac{\xb\trsps\xb}{2}\Biggl(
        \frac{
            W_{\xi-1,\frac{2a-N}{4}+1}\bigl(\frac{\xb\trsps\xb}{2}\bigr)
        }{
            W_{\xi,\frac{2a-N}{4}}\bigl( \frac{\xb\trsps\xb}{2}\bigr)
        }\rule{40pt}{0pt}\\
    \IEEEeqnarraymulticol{1}{r}{
        -\frac{
            W_{\xi-\frac{1}{2},\frac{2a-N+2}{4}}^{2}\bigl(
                \frac{\xb\trsps\xb}{2}
            \bigr)
        }{
            W_{\xi,\frac{2a-N}{4}}^{2}\bigl(\frac{\xb\trsps\xb}{2}\bigr)
        }%
        \Biggr) \Biggr] .
    }
    \IEEEyesnumber\label{eqn:MMSE_v1}%
\end{IEEEeqnarray*}%
Finally, by using the change of variable $t=\frac{\xb\trsps\xb}{2}$, we obtain
\begin{IEEEeqnarray*}{-l}
    \mathrm{MMSE}=
    \frac{\varGamma(2a)}{\varGamma(a)\,\varGamma{\left(\frac{N}{2}\right)}}
    \\
    \smash[b]{
        \times\!\int_{0}^{\infty}\limits
        \Biggl(
            W_{\xi-1,\frac{2a-N}{4}+1}(t)
            -\frac{
                W_{\xi-\frac{1}{2},\frac{2a-N+2}{4}}^{2}(t)
            }{
                W_{\xi,\frac{2a-N}{2}}(t)
            }
        \Biggr)
        t^{\frac{2a+N-2}{4}}e^{-\frac{t}{2}}\diff t.
    }\\
    \IEEEyesnumber\label{eqn:MMSE_v2} 
\end{IEEEeqnarray*}

\subsection{Expression of the TBCRB}


In the present case, \eqref{eqn:TBCRB_general}--\eqref{eqn:Qphi_x_CR} reduce to
\begin{equation}
    \mathrm{TBCRB} =
    \Expct*{\xb}{
        \frac{1}{
            F_{\xb}
        }
    },
\end{equation}
where $F_{\xb}$ denotes the posterior Fisher information, that is
\begin{IEEEeqnarray}{/l}
    F_{\xb}\eqdef
    \Expct*{\theta|\xb}{
        \Bigl(
            \frac{\partial\ln p(\theta|\xb)}{\partial\theta}
        \Bigr)^{2}
    }
    =-\Expct[\bigg]{\theta|\xb}{
        \frac{\partial^{2}\ln p(\theta\given\xb)}{\partial^{2}\theta}
    },
    \label{eqn:def_Fx}
\end{IEEEeqnarray}
and can be expressed, similarly as in \eqref{eqn:def_BFIM}, as
\begin{IEEEeqnarray*}{rCl}
    F_{\xb}&=&
    (a-1-\frac{N}{2})\,\Expct*{\theta|\xb}{\theta^{-2}} \\
    & & +\>(\xb\trsps\xb)\,\Expct*{\theta|\xb}{\theta^{-3}}
        +(a-1)\,\Expct*{\theta|\xb}{(1-\theta)^{-2}}.\qquad
    \IEEEyesnumber\label{eqn:Fx}
\end{IEEEeqnarray*}%
By the same approach as that leading to \eqref{eqn:MMSEE} and \eqref{eqn:E_theta2}, it is possible to obtain the expressions of
$\Expct{\theta|\xb}{\theta^{-2}}$,
$\Expct{\theta|\xb}{\theta^{-3}}$ and
$\Expct*{\theta|\xb}{(1-\theta)^{-2}}$.
Finally, considering again the change of variable $t=\frac{\xb\trsps\xb}{2}$, we obtain the following expression for the TBCRB:
\begin{IEEEeqnarray*}{-l}
    \mathrm{TBCRB}= \\
    \,
    \frac{\varGamma(2a)}{\varGamma(a)\,\varGamma{\left( \frac{N}{2}\right)} } 
    \int\limits_{0}^{\infty }\frac{ W_{\xi,\frac{2a-%
                N}{4}}^{2}(t)\, t^{\frac{2a+N-2}{4}} e^{-%
            \frac{t}{2}}}{
        \begin{IEEEeqnarraybox}[\renewcommand{\IEEEeqnarraymathstyle}{\textstyle}][c]{l}
        \Bigl(
        \left( a-1-\frac{N}{2}\right) W_{\xi+1,\frac{2a-N}{4}-1}(t) \\
        \quad{+}\>2\sqrt{t}\, W_{\xi+\frac{3}{2},\frac{2a-N-6}{4}}(t) \\
        \quad{+}\> ( a-1) \frac{\varGamma(a-2) }{\varGamma(a) }\, 
        t\, W_{\xi+2,\frac{2a-N}{4}}(t)%
        \Bigr)%
        \end{IEEEeqnarraybox}%
    }\diff t.
    \IEEEyesnumber \label{eqn:TBCRB}
\end{IEEEeqnarray*}

\subsection{Asymptotic results}

In this section, we aim at proving that, for the problem under study, the posterior p.d.f. $p(\theta\given\xb)$ asymptotically (as $N\rightarrow\infty$) has the form \eqref{eqn:posterior_gauss}. It consequently makes it possible to determine an asymptotically efficient estimator (in the sense of Definition \ref{def:efficiency}), which is shown to be any estimator that is asymptotically equivalent to the maximum likelihood estimator (including the MAP estimator, in particular).

\subsubsection{First asymptotic form of $p(\theta\given\xb)$}

A first asymptotic expression for the posterior p.d.f. is obtained  from \eqref{eqn:joint_pdf}, by computing its first log-derivative:
\begin{equation}
    \frac{\partial\ln p(\theta\given\xb)}{\partial\theta}
    =
    \frac{\partial\ln p(\xb,\theta)}{\partial\theta}
    = \frac{N}{2}
      \frac{\alpha\theta^{2}-\beta\theta+\gamma}{\theta^{2}(1-\theta)}
  \label{eqn:dln_posterior}
\end{equation}
where $\alpha$, $\beta$ and $\gamma$ appeared in the expression of the MAP estimator \eqref{eqn:MAPE_N<>4(a-1)}. Since we aim at analyzing the behavior of $p(\theta\given\xb)$ as $N\rightarrow+\infty$, let us consider the case where $N\neq 4(a-1)$, i.e., $\alpha\neq 0$. Then, \eqref{eqn:dln_posterior} reduces to
\begin{equation}
\frac{\partial\ln p(\theta\given\xb)}{\partial\theta}
= \frac{N}{2}
\frac{\alpha(\theta-\widehat{\theta}_{1})(\theta-\widehat{\theta}_{2})}{\theta^{2}(1-\theta)}
\label{eqn:dln_posterior_bis}
\end{equation}
where $\widehat{\theta}_{1}\eqdef\widehat{\theta}^{\mathsf{MAP}}(\xb)=(\beta-\sqrt{\beta^{2}-4\alpha\gamma})/(2\alpha)$ (as in \eqref{eqn:MAPE_N<>4(a-1)}), and 
\begin{equation} 
    \widehat{\theta}_{2}\eqdef
    \frac{\beta+\sqrt{\beta^{2}-4\alpha\gamma}}{2\alpha}.
    \label{eqn:def_theta2}
\end{equation}

Let us introduce some notations we use in the sequel. Let $z_{1},\ldots,z_{N}$ be a sequence of $N$ real random variables, and $v(\cdot),w(\cdot)$ two real-valued functions:
\begin{IEEEitemize}
    \item $v(\zb)=v(z_{1},\ldots,z_{N})\attp \ell$
    means that $v(\zb)$ tends to the value $\ell\in\Rbb$ in probability as $N$ tends to infinity, i.e., for any $\delta>0$, $\lim_{N\rightarrow+\infty}\Pr\bigl(\abs{v(\zb)-\ell}<\delta\bigr)=1$;
    
    \item $v(\zb)\aequivp w(\zb)$ means that, when $N\rightarrow +\infty$, we can write $v(\zb)=w(\zb)(1+\varepsilon(\zb))$, where $\varepsilon(\cdot)$ is a function such that $\varepsilon(\zb)\attp 0$;
    
    \item The same notations with $v(\zb)\ttms[\scriptscriptstyle] \ell$ and $v(\zb)\equivms[\scriptscriptstyle] w(\zb)$ are used where convergence in mean-square is considered instead of convergence in probability.
\end{IEEEitemize}

Before stating the main result of this section, we need two intermediate ones, which are stated in Propositions \ref{prp:ML} and \ref{prp:MAP}.

\begin{proposition}\label{prp:ML}
    The maximum likelihood estimator $\widehat{\theta}^{\mathsf{ML}}$ of $\theta$, given by $\widehat{\theta}^{\mathsf{ML}}(\xb)=(\xb\trsps\xb)/N$, asymptotically tends to $\theta$ in probability:
    \begin{equation}
        \widehat{\theta}^{\mathsf{ML}}(\xb)=\frac{\xb\trsps\xb}{N}
        \attp \theta.
        \label{eqn:thetaML_attp_theta}
    \end{equation}
\end{proposition}

\begin{proof}
    We prove the convergence in mean-square, which implies the convergence in probability. Since $\Expct{\xb|\theta}{\smash{\frac{\xb\trsps\xb}{N}}}=\theta$, we can write
    \begin{IEEEeqnarray*}{rCl}
        \Expct*{\xb,\theta}{\Bigl(\frac{\xb\trsps\xb}{N}-\theta\Bigr)^{2}}
        &=&\Expct*{\theta}{\Expct*{\xb|\theta}{\Bigl(\frac{\xb\trsps\xb}{N}-\theta\Bigr)^{2}}} \\
        &=&\Expct*{\theta}{\Var*{\xb|\theta}{\frac{\xb\trsps\xb}{N}}},
        \IEEEyesnumber\label{eqn:ML_cv_theta_ms}
    \end{IEEEeqnarray*}
    where
    \begin{equation}
        \Var*{\xb|\theta}{\frac{\xb\trsps\xb}{N}}
        = \frac{1}{N^{2}}\sum_{n=1}^{N}\Var*{\xb|\theta}{x_{n}^{2}}
        = \frac{\theta^{2}}{N}
        \Var*{\xb|\theta}{\frac{x_{n}^{2}}{\theta}}.
    \end{equation}    
    Since $x_{1},\ldots,x_{N}$ are i.i.d. such that $x_{n}\given\theta\sim\mathcal{N}(0,\theta)$, then $(x_{n}/\sqrt{\theta})\given\theta\sim\mathcal{N}(0,1)$. Thus, $x_{n}^{2}/\theta$ follows a chi-squared distribution with 1 degree of freedom, which implies $\Var*{\xb|\theta}{x_{n}^{2}/\theta}=2$. Consequently, \eqref{eqn:ML_cv_theta_ms} becomes
    \begin{equation}
        \Expct*{\xb,\theta}{\Bigl(\frac{\xb\trsps\xb}{N}-\theta\Bigr)^{2}}
        = \Expct*{\theta}{\frac{2\theta^{2}}{N}}
        = \frac{2}{N}\Expct{\theta}{\theta^{2}},
    \end{equation}
    which tends to 0 as $N\rightarrow\infty$.
\end{proof}

\begin{proposition}\label{prp:MAP}
    The MAP estimator $\widehat{\theta}^{\mathsf{MAP}}$ asymptotically behaves as the maximum likelihood (ML) estimator $\widehat{\theta}^{\mathsf{ML}}$, that is
    \begin{equation}
        \widehat{\theta}^{\mathsf{MAP}}(\xb)\aequivp
        \widehat{\theta}^{\mathsf{ML}}(\xb)=\frac{\xb\trsps\xb}{N},
        \label{eqn:thetaMAP_aequivp_thetaML}
    \end{equation}
    and
    \begin{equation}
        \widehat{\theta}_{2}\attp 1.
        \label{eqn:theta2_aequivp_1}
    \end{equation}
\end{proposition}

\begin{proof}
    In \eqref{eqn:MAPE_N<>4(a-1)}, we have
    \begin{IEEEeqnarray}{rCl}
        \IEEEyesnumber\IEEEyessubnumber
        \alpha &=& 1-\frac{4(a-1)}{N} \underset{N\rightarrow\infty}{\longrightarrow}1,
        \label{eqn:equiv_alpha}
    \end{IEEEeqnarray}
    \begin{IEEEeqnarray}{rCl}
        \beta = 1-\frac{2(a-1)}{N}+\frac{\xb\trsps\xb}{N}
        &=&1+\frac{\xb\trsps\xb}{N}\Bigl(1-\frac{2(a-1)}{\xb\trsps\xb}\Bigr)
        \IEEEnonumber\\
        &\underset{\mathclap{N\rightarrow\infty}}{\equivp}&\mspace{12mu}
        1+\frac{\xb\trsps\xb}{N}
        \IEEEyessubnumber\label{eqn:equiv_beta}
    \end{IEEEeqnarray}
    since we can deduce from Proposition \ref{prp:ML} that $(\xb\trsps\xb)^{-1}\attp 0$, and
    \begin{IEEEeqnarray}{rCl}
        \gamma &=& \frac{\xb\trsps\xb}{N}. \IEEEyessubnumber
    \end{IEEEeqnarray}
    Therefore, we can show that
    \begin{equation}
        \beta^{2}-4\alpha\gamma\aequivp
        \Bigl(1-\frac{\xb\trsps\xb}{N}\Bigr)^{2},
        \label{eqn:equiv_Delta}
    \end{equation}
    which leads, after plugging \eqref{eqn:equiv_alpha}, \eqref{eqn:equiv_beta} and \eqref{eqn:equiv_Delta} into \eqref{eqn:MAPE_N<>4(a-1)}, to
    \begin{equation}
        \widehat{\theta}^{\mathsf{MAP}}(\xb)\aequivp
        \frac{1+\frac{\xb\trsps\xb}{N}-\abs{1-\frac{\xb\trsps\xb}{N}}}{2}.
        \label{eqn:equiv_thetaMAP_v0}
    \end{equation}
    Yet, by definition of \eqref{eqn:thetaML_attp_theta}, we have
    \begin{equation}
        \forall\delta>0,\ 
        \lim_{N\rightarrow\infty}
        \Pr\Bigl(\abs[\Big]{\frac{\xb\trsps\xb}{N}-\theta}<\delta\Bigr) = 1,
    \end{equation}
    which implies, in particular, that $\lim_{N\rightarrow\infty}
    \Pr(1-\frac{\xb\trsps\xb}{N}>0) = 1$. Hence, we obtain \eqref{eqn:thetaMAP_aequivp_thetaML} from \eqref{eqn:equiv_thetaMAP_v0}. Similarly, after plugging \eqref{eqn:equiv_alpha}, \eqref{eqn:equiv_beta} and \eqref{eqn:equiv_Delta} into \eqref{eqn:def_theta2}, we obtain \eqref{eqn:theta2_aequivp_1}.
\end{proof}

The following proposition makes up the main result of this section.

\begin{proposition}\label{prp:asymptotic_form_posterior_1}
    A first asymptotic form of the posterior p.d.f. $p(\theta\given\xb)$ is given by
    \begin{equation}
        p(\theta\given\xb)\aequivp
        \nu(\xb)\exp\biggl[
            \frac{N}{2}\biggl(
                \ln\biggl(\frac{1}{\theta}\frac{\xb\trsps\xb}{N}\biggr)
                -\frac{1}{\theta}\frac{\xb\trsps\xb}{N}
            \biggr)
        \biggr],
        \label{eqn:posterior_aequivp_v1}
    \end{equation}
    where $\nu(\cdot)$ is a normalizing function of $\xb$ only, ensuring that
    $\int_{\mathcal{S}_{\Theta|\xb}}p(\theta\given\xb)\diff\theta=1$.
\end{proposition}

\begin{proof}
    Using \eqref{eqn:thetaMAP_aequivp_thetaML} and \eqref{eqn:theta2_aequivp_1}, we have, from \eqref{eqn:dln_posterior_bis},
    \begin{equation}
        \frac{\partial\ln p(\theta\given\xb)}{\partial\theta}\aequivp
        \frac{N}{2}\frac{
            (\theta-\frac{\xb\trsps\xb}{N})(\theta-1)
        }{\theta^{2}(1-\theta)}
        = \frac{N(\frac{\xb\trsps\xb}{N}-\theta)}{2\theta^{2}}
        \label{eqn:dln_posterior_aequivp}
    \end{equation}
    After integrating both sides w.r.t. $\theta$, we obtain
    \begin{equation}
        \ln p(\theta\given\xb) \aequivp
        -\frac{\xb\trsps\xb}{2\theta}-\frac{N}{2}\ln\theta+C(\xb), 
        \label{eqn:ln_posterior_aequivp}
    \end{equation}
    where $C(\xb)$ denotes an arbitrary function of $\xb$. By setting, $\nu(\xb)=(\frac{\xb\trsps\xb}{N})^{-N/2}\exp C(\xb)$, \eqref{eqn:ln_posterior_aequivp} leads to \eqref{eqn:posterior_aequivp_v1}.
\end{proof}

As it can be noticed, the asymptotic form obtained for $p(\theta\given\xb)$ in \eqref{eqn:posterior_aequivp_v1} is not the same as in \eqref{eqn:posterior_gauss}. However, as shown in the next section, a second asymptotic form of $p(\theta\given\xb)$ that has the required form \eqref{eqn:posterior_gauss} can be obtained, starting from \eqref{eqn:posterior_aequivp_v1}.

\subsubsection{Second asymptotic form of $p(\theta\given\xb)$}

The main result of this section is stated as follows.

\begin{proposition} \label{prp:asymptotic_form_posterior_2}
    A second asymptotic form of the posterior p.d.f. $p(\theta\given\xb)$ is given by
    \begin{equation}
        p(\theta\given\xb)\aequivp
        \xi(\xb)\exp\biggl[
            -\frac{1}{2\frac{2}{N}\bigl(\frac{\xb\trsps\xb}{N}\bigr)^{2}}
            \Bigl(\theta-\frac{\xb\trsps\xb}{N}\Bigr)^{2}
        \biggr],
        \label{eqn:posterior_aequivp_v2}
    \end{equation}
    where $\xi(\cdot)$ is a normalizing function of $\xb$ only, ensuring that
    $\int_{\mathcal{S}_{\Theta|\xb}}p(\theta\given\xb)\diff\theta=1$.
\end{proposition}

\begin{proof}
    Let us define
    \begin{equation}
        h_{\xb}(\theta) \eqdef
        \ln\biggl(\frac{1}{\theta}\frac{\xb\trsps\xb}{N}\biggr)
        -\frac{1}{\theta}\frac{\xb\trsps\xb}{N}.
    \end{equation}
    A Taylor expansion of $h_{\xb}(\cdot)$ at the vicinity of
    $\widehat{\theta}^{\mathsf{ML}}(\xb)=\frac{\xb\trsps\xb}{N}$ gives, after setting $d\theta=\theta-\frac{\xb\trsps\xb}{N}$,
    \begin{IEEEeqnarray*}{l}
        h_{\xb}\biggl(\frac{\xb\trsps\xb}{N}+d\theta\biggr)\\
        \quad=
        h_{\xb}\biggl(\frac{\xb\trsps\xb}{N}\biggr)
        +d\theta\frac{
                    \diff h_{\xb}(\theta)
                }{
                    \diff\theta
                }\bigg|_{\frac{\xb\trsps\xb}{N}} 
        \mspace{-15mu}+\frac{d\theta^{2}}{2}\frac{
            \diff^{2} h_{\xb}(\theta)
        }{
            \diff\theta^{2}
        }\bigg|_{\frac{\xb\trsps\xb}{N}} \mspace{-15mu}+ o(d\theta^{2}) \\
    \quad=-1 + d\theta\times 0 -\frac{d\theta^{2}}{2\bigl(\frac{\xb\trsps\xb}{N}\bigr)^{2}}+ o(d\theta^{2}) \\
    \quad= -1-\frac{(\theta-\frac{\xb\trsps\xb}{N})^{2}}{2\bigl(\frac{\xb\trsps\xb}{N}\bigr)^{2}}+ o\Bigl(\Bigl(\theta-\frac{\xb\trsps\xb}{N}\Bigr)^{2}\Bigr). \IEEEyesnumber \label{eqn:taylor_exp_hx}
    \end{IEEEeqnarray*}
    This yields, for $\theta\in\intervaloo{\frac{\xb\trsps\xb}{N}-d\theta}{\frac{\xb\trsps\xb}{N}+d\theta}$,
    \begin{equation}
        p(\theta\given\xb)
        \underset{N\rightarrow\infty,d\theta\rightarrow 0}{\overset{P}{\sim}}
        \xi(\xb)\exp\biggl[
            -\frac{1}{2\frac{2}{N}\bigl(\frac{\xb\trsps\xb}{N}\bigr)^{2}}
            \Bigl(\theta-\frac{\xb\trsps\xb}{N}\Bigr)^{2}
        \biggr],
        \label{eqn:local_posterior_aequivp}
    \end{equation}
    where $\xi(\xb)=\nu(\xb)\exp[-N/2]$. It is worth noting that \eqref{eqn:local_posterior_aequivp} is only a local approximation of the posterior p.d.f. (at the vicinity of $\frac{\xb\trsps\xb}{N}$). However, we can deduce from Proposition \ref{prp:ML} that
    $\frac{2}{N}\bigl(\frac{\xb\trsps\xb}{N}\bigr)^{2}\attp 0$, i.e., for any $D>0$,
    \begin{equation}
        \lim_{N\rightarrow\infty} \Pr\biggl(
            \frac{2}{N}\Bigl(\frac{\xb\trsps\xb}{N}\Bigr)^{2}
            <\frac{d\theta^{2}}{D^{2}}
        \biggr)
        =1,
    \end{equation}
    which is equivalent to
    \begin{equation}
        \lim_{N\rightarrow\infty} \Pr\Bigl(
            \mathcal{I}_{N,D}(\xb)\subset\intervalcc[\Big]{\frac{\xb\trsps\xb}{N}-d\theta}{\frac{\xb\trsps\xb}{N}+d\theta}
        \Bigr)=1,
        \label{eqn:intervals_subsets_P1}
    \end{equation}
    where 
    $\mathcal{I}_{N,D}(\xb)\eqdef \intervalcc{
        \frac{\xb\trsps\xb}{N}-D\sqrt{\frac{2}{N}}\frac{\xb\trsps\xb}{N}
    }{
        \frac{\xb\trsps\xb}{N}+D\sqrt{\frac{2}{N}}\frac{\xb\trsps\xb}{N}
    }\cap\intervalcc{0}{1}$. In addition, \eqref{eqn:local_posterior_aequivp} implies that, for sufficiently large $D$,
    \begin{equation}
        \Pr(\theta\notin\mathcal{I}_{N,D}(\xb)\given\xb)=0.
        \label{eqn:pr_theta_notin_I_ND}
    \end{equation}
    Thus, \eqref{eqn:intervals_subsets_P1} and \eqref{eqn:pr_theta_notin_I_ND} imply that \eqref{eqn:local_posterior_aequivp} becomes \eqref{eqn:posterior_aequivp_v2}.
\end{proof}

\subsubsection{Asymptotic efficiency}

Finally, Proposition \ref{prp:asymptotic_form_posterior_2} shows that asymptotically, the posterior p.d.f. has the form \eqref{eqn:posterior_gauss}, with $v(\xb) = \frac{2}{N}\bigl(\frac{\xb\trsps\xb}{N}\bigr)^{2}$, and $\widehat{\theta}(\xb) = \widehat{\theta}^{\mathsf{ML}}(\xb)$. We notice that $v(\xb)$ indeed depends on $\xb$, thus we can deduce, from discussion in Section \ref{ssec:equality_BCRBS}, that the BCRB and the TBCRB are not equivalent asymptotically. In addition, due to Proposition \ref{prp:MAP}, both the ML and the MAP estimators are asymptotically efficient, and due to \eqref{eqn:LB_MSE_BCRBs}, so is the MMSE estimator. Last but not least, an asymptotic expression for the TBCRB can be obtained, that is nothing else than the ECRB (see Section \ref{ssec:exple_known_results}), as stated in the following proposition.
\begin{proposition} \label{prp:TBCRB_ECRB}
    Asymptotically, the TBCRB is equivalent to the ECRB, i.e., 
    \begin{equation}
        \mathrm{TBCRB}\aequivp\mathrm{ECRB},
        \label{eqn:TBCRB_aequivp_ECRB}
    \end{equation}
    where the ECRB is given by \eqref{eqn:ECRB_exple}.
\end{proposition}
\begin{proof}
    Following from \eqref{eqn:TBCRB_general}--\eqref{eqn:Qphi_x_CR}, \eqref{eqn:expr_v}, and $g(\theta) = \theta$, we have
    \begin{equation}
        \mathrm{TBCRB} \aequivp \Expct{\xb}{v(\xb)},\label{eqn:TBCRB_aequivp_v1}
    \end{equation}
    where $v(\xb) = \frac{2}{N}\bigl(\frac{\xb\trsps\xb}{N}\bigr)^{2}$. Since $x_{1},\ldots,x_{N}$ are i.i.d. such that $x_{n}\given\theta\sim\mathcal{N}(0,\theta)$, then $(x_{n}/\sqrt{\theta})\given\theta\sim\mathcal{N}(0,1)$. Thus, for a given $\theta$, the random variable $t\given\theta\eqdef(\xb\trsps\xb)/\theta$ follows a chi-squared distribution with $N$ degrees of freedom: $t\given\theta\sim\chi^{2}_{N}$. Consequently,
    \begin{IEEEeqnarray*}{rCl}
        \Expct{\xb}{(\xb\trsps\xb)^{2}}
        = \Expct{\xb,\theta}{(\xb\trsps\xb)^{2}}
        &=& \Expct[\bigg]{\theta}{\theta^{2}\Expct[\bigg]{\xb\given\theta}{
                \biggl(\frac{\xb\trsps\xb}{\theta}\biggr)^{2}
            }} \\
        &=& \Expct*{\theta}{\theta^{2}\Expct*{t\given\theta}{t^{2}}}.
        \IEEEyesnumber
    \end{IEEEeqnarray*}
    The moment-generating function of $t$ is given by $M(t)\eqdef (1-2t)^{-N/2}$, and we have
    \begin{equation}
        \Expct*{t\given\theta}{t^{2}}=\frac{\diff^{2}M(t)}{\diff t^{2}}\bigg|_{t=0} = N(N+2),
    \end{equation}
    which leads to $\Expct{\xb}{(\xb\trsps\xb)^{2}}=N(N+2)\Expct*{\theta}{\theta^{2}}$. After plugging this relation into \eqref{eqn:TBCRB_aequivp_v1}, we obtain
    \begin{equation}
        \mathrm{TBCRB}\aequivp\frac{2}{N}\Bigl(1+\frac{2}{N}\Bigr)\Expct{\theta}{\theta^{2}},
    \end{equation}
    and noticing that $\frac{2}{N}\bigl(1+\frac{2}{N}\bigr)\Expct{\theta}{\theta^{2}}\aequivp\frac{2}{N}\Expct{\theta}{\theta^{2}}$ yields \eqref{eqn:TBCRB_aequivp_ECRB}.
\end{proof}

\subsection{Numerical results}\label{ssec:simulations}

In this section, we provide simulation results that illustrate the theoretical ones from the previous sections. these results appeared in \cite{BFOC19}, and are given here for sake of completeness.

Figure \ref{fig1} shows root MSEs (RMSEs) on the estimation of $\theta=\sigma^{2}$ of the MAP and MMSE estimators, as well as the BCRB, the TBCRB and the ECRB as functions of $N$, for $N=2^{n},n\in\{1,\ldots,13\}$, and the shape parameter $a=3$ in the prior density \eqref{eqn:prior}. The RMSEs of the MAP and the MMSE estimators are computed from \eqref{eqn:MAPE_N<>4(a-1)}--\eqref{eqn:MAPE_N=4(a-1)} and \eqref{eqn:MMSEE} respectively, and averaged through 20,000 Monte-Carlo trials using Matlab, which is not able to compute the Whittaker functions $W_{\mu,\nu}(z)$ appearing in \eqref{eqn:MMSEE} for $N\geq 110$. As for the theoretical square-root of the MMSE, it was computed from \eqref{eqn:MMSE_v2} using Mathematica, which is able to compute it for larger values of $N$. Consequently, the comparison between the empirical and theoretical MMSEs is only available for $N\leq 64$, where we notice a perfect match. Figure \ref{fig1} illustrates a number of points discussed in the previous sections of this paper. In particular, Proposition \ref{prp:asymptotic_form_posterior_2} and its consequences, i.e., both the MAP and the MMSE estimators are asymptotically efficient, according to Definition \ref{def:efficiency}, that is relatively to the TBCRB rather than the BCRB. For large values of $N$, we notice that i) a clear gap separates the TBCRB from the BCRB, and ii) the TBCRB and the ECRB indeed tend to the same value, which illustrates Proposition \ref{prp:TBCRB_ECRB}, and validates the asymptotic expression of $p(\theta\given\xb)$ \eqref{eqn:posterior_aequivp_v2}. Finally, as $N$ decreases, both the MAP and the MMSE estimators' RMSEs tend to the prior standard deviation $\sigma_{\pi}=1/(2\sqrt{7})$ on the one hand, while the BCRB and the TBCRB tend to the same value, since the prior information dominates, and the posterior variance does practically not depend on the observations $\xb$ anymore. 

\begin{figure}[t!]
    \centering
    \includegraphics[width=\linewidth]{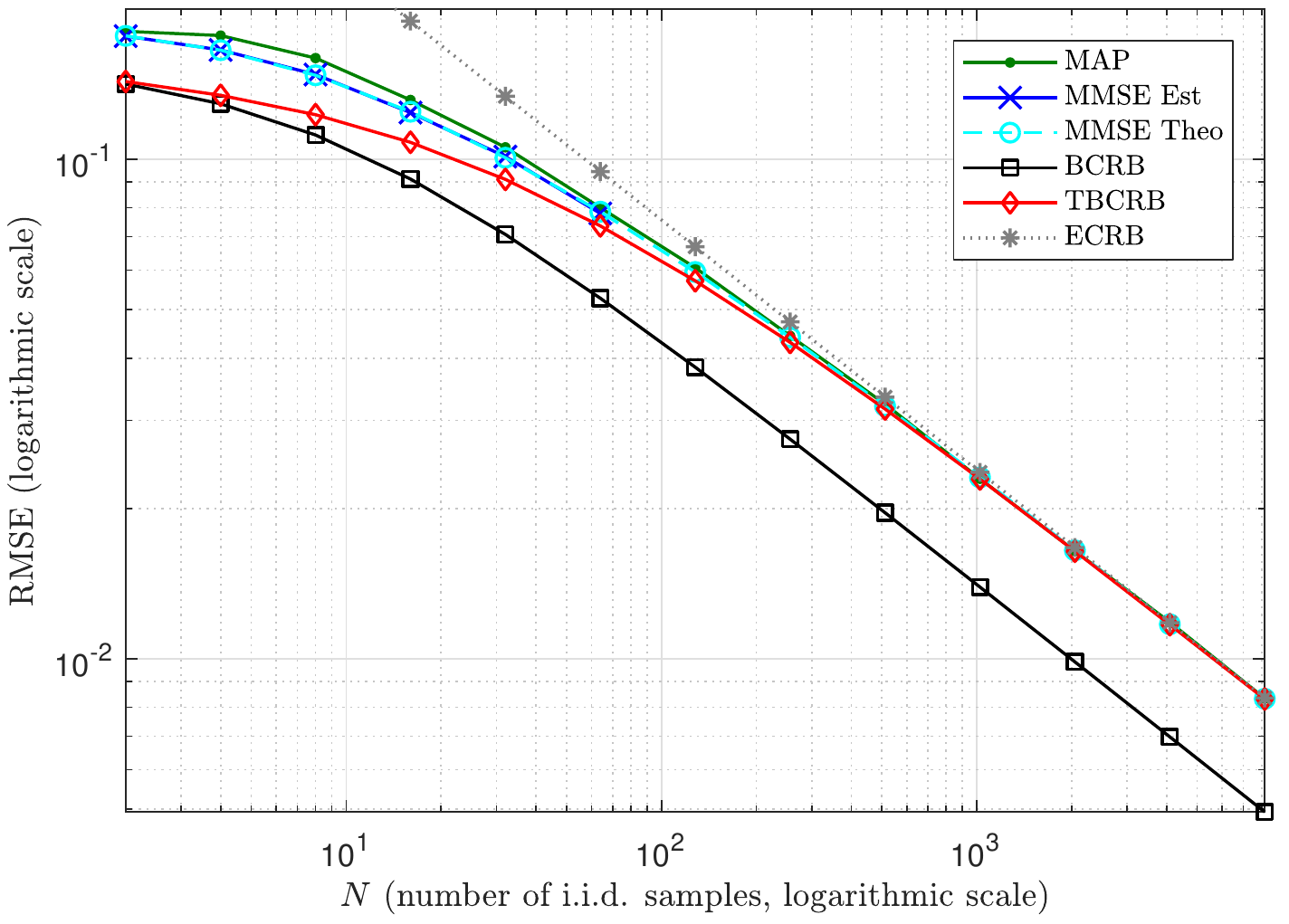}
    \caption{
        RMSE of the MAP (green dots) and the MMSE (estimated via Monte-Carlo: blue ``$\times$'' signs; theoretical from \eqref{eqn:MMSE_v2}: cyan circles, dashed line) estimators, with BCRB (black squares), TBCRB (red diamonds) and ECRB (gray ``$*$'' signs) versus $N$, in estimating $\theta=\sigma^{2}$.
    }
    \label{fig1}
\end{figure}

\section{Concluding remarks} \label{sec:conclusion}

Any Bayesian lower bound on the MSE in the Weiss-Weinstein family has an alternative form, which turns out to be at least as tight as the standard form. We have given a proof for the case of the estimation of a vector parameter $\thetab$, or any known vector function $\gb(\thetab)$ of a vector parameter. 
The tighter forms of BLBs give rise to an update in the definition of efficient estimation in the Bayesian framework, and make it possible to derive new conditions for efficiency. A sufficient condition for the standard and the tighter forms to differ is that the posterior autocorrelation matrix $\Qb_{\phib\given\xb}$ of the BLB-generating functions $\phib(\xb,\thetab)$, as well as the posterior intercorrelation matrix $\Rb_{\gb\phib\given\xb}$, do depend upon $\xb$, the given set of observations. This condition is likely to be met, except when few observations are available. In the case of the estimation of a scalar quantity $g(\theta)$, we have shown that an efficient estimator can be found if and only if the posterior distribution $p(\theta\given\xb)$ has a particular form in the exponential family of distributions. Conversely, we have derived conditions for finding efficient estimators in the case of exponential family models with conjugate prior distribution. Finally, we illustrated the precision gain and relevance of the tighter forms of BLBs through a noteworthy example, for which the standard form of the BCRB is known to not be tight. We have shown that, in the asymptotic regime, Bayesian estimators like the MAP or the MMSE attain the TBCRB, and then are efficient. These results were validated by numerical simulations.

\bibliographystyle{IEEEtran}
\bibliography{mybib}

\end{document}